\documentclass[11pt]{article}

\usepackage[ruled]{algorithm2e}
\SetAlFnt{\small}
\SetAlCapFnt{\small}
\SetAlCapNameFnt{\small}
\SetAlCapHSkip{0pt}
\IncMargin{-\parindent}

\usepackage{algorithmic}
\usepackage{nicefrac}

\usepackage[
  bookmarks=true,
  bookmarksnumbered=true,
  bookmarksopen=true,
  pdfborder={0 0 0},
  breaklinks=true,
  colorlinks=true,
  linkcolor=black,
  citecolor=black,
  filecolor=black,
  urlcolor=black,
]{hyperref}
\usepackage{graphicx}
\usepackage{xcolor}
\usepackage[margin=1in]{geometry}
\usepackage{amssymb}
\usepackage{amsmath}
\usepackage[square]{natbib}
\setcitestyle{numbers}
\usepackage{amsthm}
\usepackage[capitalize]{cleveref}
\usepackage{bbm}	

\Crefname{figure}{Figure}{Figures}
\Crefname{equation}{Equation}{Equations}

\theoremstyle{plain}
\newtheorem{theorem}{Theorem}[section]
\newtheorem{corollary}[theorem]{Corollary}
\newtheorem{lemma}[theorem]{Lemma}
\newtheorem{proposition}[theorem]{Proposition}
\newtheorem{claim}[theorem]{Claim}
\newtheorem{observation}[theorem]{Observation}
\newtheorem{open}[theorem]{Open Question}

\theoremstyle{definition}
\newtheorem{definition}[theorem]{Definition}
\newtheorem{example}[theorem]{Example}

\newcommand{\com}{F_{\text{co-mon}}}
\newcommand{\eqrev}{F_{\text{eq-rev}}}
\newcommand{\Exp}{\operatorname{Exp}}
\newcommand{\ALG}{\operatorname{ALG}}
\newcommand{\OPT}{\operatorname{OPT}}
\newcommand{\Mye}{\operatorname{Mye}}

\hypersetup{
  pdfauthor      = {Moshe Babaioff <moshe@microsoft.com>, Michal Feldman <michal.feldman@cs.tau.ac.il>, Yannai A. Gonczarowski <yannai@gonch.name>, Brendan Lucier <brlucier@microsoft.com>, Inbal Talgam-Cohen <italgam@cs.technion.ac.il>},
  pdftitle       = {Escaping Cannibalization? Correlation-Robust Pricing for a Unit-Demand Buyer},
}
\begin{document}

\title{Escaping Cannibalization?\texorpdfstring{\\}{ }Correlation-Robust Pricing for a Unit-Demand Buyer\thanks{A one-page abstract of this paper appeared in the Proceedings of the 21st ACM Conference on Economics and Computation (EC 2020). A one-minute flash video for this paper, which was awarded the EC 2020 Best Flash Video award, is available at: \url{https://www.youtube.com/watch?v=w3fQt3CklxI} . Feldman and Gonczarowski were supported in part by ISF grant 317/17 administered by the Israeli Academy of Sciences. Feldman was supported in part by the European Research Council (ERC) under the European Union's Horizon 2020 research and innovation program (grant agreement No.\ 866132). Gonczarowski was supported in part by ISF grant 1841/14 administered by the Israeli Academy of Sciences. Talgam-Cohen is a Taub Fellow supported by the Taub Family Foundation, and by ISF grant 336/18 administered by the Israeli Academy of Sciences.}}
\date{August 15, 2020}

\author{Moshe Babaioff\thanks{Microsoft Research | \emph{E-mail}: \href{mailto:moshe@microsoft.com}{moshe@microsoft.com}.} \and Michal Feldman\thanks{Tel Aviv University and Microsoft Research | \emph{E-mail}: \href{mailto:michal.feldman@cs.tau.ac.il}{michal.feldman@cs.tau.ac.il}.} \and Yannai A. Gonczarowski\thanks{Microsoft Research | \emph{E-mail}: \href{mailto:yannai@gonch.name}{yannai@gonch.name}. Research carried out in part while at Tel Aviv University.} \and Brendan Lucier\thanks{Microsoft Research | \emph{E-mail}: \href{mailto:brlucier@microsoft.com}{brlucier@microsoft.com}.} \and Inbal Talgam-Cohen\thanks{Technion--Israel Institute of Technology | \emph{E-mail}: \href{mailto:italgam@cs.technion.ac.il}{italgam@cs.technion.ac.il}.}}

\maketitle

\begin{abstract}
	We consider a robust version of the revenue maximization problem, where a single seller wishes to sell $n$ items to a single unit-demand buyer. In this robust version, the seller knows the buyer's marginal value distribution for each item separately, but not the joint distribution, and prices the items to maximize revenue in the worst case over all compatible correlation structures.
We devise a computationally efficient (polynomial in the support size of the marginals) algorithm that computes the worst-case joint distribution for any choice of item prices.
And yet, in sharp contrast to the additive buyer case (\citeauthor{Carroll17}, \citeyear{Carroll17}), we show that it is NP-hard to approximate the optimal choice of prices to within any factor better than $n^{1/2-\epsilon}$. 
For the special case of marginal distributions that satisfy the \emph{monotone hazard rate} property, we show how to guarantee a constant fraction of the optimal worst-case revenue using item pricing; this pricing equates revenue across all possible correlations and can be computed efficiently.  
\end{abstract}

\section{Introduction}
\label{sec:intro}
A core field of study within Algorithmic Mechanism Design is that of the design of selling mechanisms, with one of the most fundamental questions being that of revenue-maximization by a single seller, even when facing only a single buyer. The standard
setting for this question is the Bayesian setting, where the seller knows a prior distribution from which the values of the buyer for the various items are drawn, and aims to maximize her revenue in expectation over this prior distribution. 
When the seller has only a single item for sale, the optimal mechanism in such setting turns out to be a simple pricing mechanism, as established by Myerson~\cite{Myerson} and Riley and Zeckhauser~\cite{RileyZeckhauser}.
But for multiple items, even for simple valuation models such as \emph{additive} or \emph{unit-demand} valuations, and even assuming that the buyer's values are independent across items, revenue-optimal mechanisms can be quite complex~\cite{DaskalakisDT17}, hard to compute~\cite{ChenDPSY18,DaskalakisDT14}, and exhibit unintuitive behavior such as nonmonotonicity in the valuation distributions \cite{HartReny}, with the general revenue-maximizing solution even in these settings continuing to elude a complete characterization to date.
In such settings, the search for simple mechanisms, and in particular for \emph{pricing mechanisms} (i.e., deterministic mechanisms, which price items and/or bundles of items) that give good revenue guarantees, has therefore spawned many important results \cite{ChawlaHMS10,ChawlaHK07,BabaioffILW14,HartN17}.

As one might expect, there are grave impossibility results \cite{HartNisan2019Correlated,BriestK11} for the Bayesian setting when the prior distribution (that is perfectly known to the seller) over the valuations of the various items exhibits correlations between these valuations (rather than the items being independently distributed).
Nonetheless, Carroll~\cite{Carroll17} has asked whether 
the situation may become more hopeful
under a ``partially known underlying distribution'' scenario in which
the seller is only given the marginal distribution of each valuation, and wishes to maximize her guaranteed expected revenue over any possible correlated valuation distribution with the given marginals.

In his pioneering paper, among other contributions, Carroll~\cite{Carroll17} considers the scenario of a single seller with a number of items to sell to a single buyer with an \emph{additive} valuation, where the seller knows the distribution of the buyer's valuation for each good separately.
Remarkably, 
the mechanism that provides the highest ``worst case'' (expected) revenue guarantee across all such correlations is an exceedingly simple pricing mechanism: it simply prices each item separately according to its optimal take-it-or-leave-it price {\`a} la Myerson / Riley and Zeckhauser.  Interestingly, since this pricing mechanism prices only single items and not bundles of items, it has the appealing property that its expected revenue is independent of the correlation structure. One might therefore take an intuitive message that in the absence of knowledge about such correlations, one should opt for a mechanism whose revenue is agnostic of these correlations.  This message is further echoed in recent extensions to budgeted buyers~\cite{GravinL18} and optimal contract design\footnote{In this robust contract design problem, the hidden information is not correlations but rather the higher moments of the reward distributions for the principle~\cite{DuttingRT19}, however the message remains: the (expected) revenue of the (worst-case) optimal mechanism/contract is agnostic of the information that is hidden from the mechanism/contract designer (even though for other, nonoptimal, mechanisms/contracts, this information is needed to compute the revenue).}~\cite{DuttingRT19}.  As such, one might naturally wonder whether such a principle of agnosticism holds more generally, even if only for approximate revenue-maximization; or, failing that, whether one can always find simple pricing mechanisms (correlation-agnostic or not) that even only approximate the (worst-case) optimal revenue.

In this paper we study \emph{revenue-maximizing pricing} in a correlation-robust setting where a seller with multiple items faces a single \emph{unit demand} buyer, and in particular consider the above question through the lens of this setting. The analysis of this setting was posed as an open problem by Gravin and Lu \cite{GravinL18}, who in particular also explicitly posed the question of the tractability of finding a solution.  We first ask: does the optimal correlation-robust mechanism take the form of a correlation-agnostic pricing mechanism that can be computed efficiently? (As is often done, we use ``can be computed efficiently'' to formalize the amorphic ``simple'' from the above discussion.) As it turns out, the answer is no; we present examples in which even the optimal choice of item prices\footnote{A pricing mechanism (any deterministic mechanism) for a unit-demand buyer without loss of generality simply sets a price for each item (being the lowest price of any bundle containing the item), and need not offer any bundles since a unit-demand buyer has no value for more than one item.} (a lower benchmark than the optimal mechanism) leads to an unboundedly higher revenue than any correlation-agnostic choice of prices as $n$ grows large.  This negative answer gives rise to a new challenge: identifying the revenue guarantee given the (worst-case) optimal pricing.  Our first main result addresses this challenge and shows that given the optimal pricing, and in fact given \emph{any} pricing, the revenue guarantee of that pricing can be efficiently calculated:
\begin{theorem}[See \cref{thm:alg-is-opt}]
Given $n$ discrete (marginal) distributions $F_1,\ldots,F_n$ from which a unit-demand buyer's valuations of $n$ respective items are drawn, and given respective prices $p_1,\ldots,p_n$ for these items, the correlation among the $n$ given distributions that gives the lowest (expected) revenue from the buyer, as well as that revenue itself, can be computed in polynomial time.
\end{theorem}

In other words, the high-dimensional problem of finding a revenue-minimizing correlation given prices is tractably solvable. What about the lower-dimensional problem of finding a pricing that well approximates the worst-case optimal revenue?

A major hurdle in finding high-revenue pricings for a unit-demand buyer, a hurdle that does not arise in the additive case, is that of \emph{cannibalization}, whereby one item that is offered for sale \emph{cannibalizes} from the revenue of another item.\footnote{This cannibalization issue is common to the literature on assortment planning, where typically the prices are fixed and the decision-maker must choose which items to make available; see~\cite{kok2008assortment} for a survey.}
Consider for example two items, one priced at a low price, say, \$1, and the other at a high price, say, \$1,000,000. Say that the buyer has realized values \$1.5 for the first item and \$1,000,000.25 for the second item. Such a buyer would opt for buying the first item, resulting in a revenue of only \$1, as the buyer's utility from that item (at that price) is slightly higher. Since the buyer has unit demand, she would therefore not buy the second item. It is not hard to see that for this particular value realization, pricing the first item at infinity (i.e., not offering it for sale at all) leads to higher revenue.
When the valuations are drawn from an underlying distribution, the extent to which cannibalization affects the seller's revenue is of course intimately connected to the correlation between the values of different items.

In the Bayesian setting, if item valuations are independently drawn, then while cannibalization manifests to some extent, it turns out that simple pricings can still achieve a constant approximation to the optimal revenue for any number of items~\cite{ChawlaHMS10,ChawlaHK07}. 
But correlations between item values can potentially amplify cannibalization.  For example, one might correlate values so that the buyer often has just slightly higher utility from low-priced items than from high-priced items.
Such a ``bad'' correlation could depend crucially on the specific choice of item prices, though: they determine which values are above-the-price, and the utility given by each of these values.
Therefore, avoiding excessive cannibalization is in some sense even more challenging in the correlation-robust setting, as in this setting the worst-case correlation is effectively tailored to maximally cannibalize the revenue of the chosen prices (rather than given in advance, which gives the seller a chance to price in a way that mitigates the cannibalization caused by a specific correlation).

Keeping cannibalization under control, a challenge that is completely absent from the additive setting, is indeed the main challenge in our unit-demand setting.
As we show, finding prices that best overcome this challenge is as hard as getting a relatively good approximation to the notoriously hard problem of finding the maximal independent set in a graph.
Hence, it is not only hard to find prices that ``thread the needle'' by fighting cannibalization ``just right'' to optimize the (worst-case) revenue among all prices, but it is in fact NP-hard to find prices that even coarsely approximate the revenue among all prices to any reasonable extent:

\begin{theorem}[See \cref{thm:hardness}]
The \emph{max-min robust pricing} problem is NP-hard: given $n$ (marginal) distributions $F_1,\ldots,F_n$, each described using at most $\mathrm{poly}(n)$ bits, from which a unit-demand buyer's valuations of $n$ respective items are drawn, compute respective prices $p_1,\ldots,p_n$ for these items that maximize, among all possible choices of prices, the guaranteed (expected) revenue from the buyer over any correlation among the $n$ given distributions. Furthermore, it is NP-hard to find prices that even approximate this guaranteed revenue up to a factor of $n^{1/2-\varepsilon}$ for any $\varepsilon>0$.
\end{theorem}

This theorem strongly indicates that a ``clean'' characterization of (approximately) optimal robust prices is unlikely to exist. The theorem also immediately implies the same lower bound also for finding prices that approximate the more ambitious benchmark of the (worst-case) optimal revenue from any (not necessarily pricing) mechanism. To the best of our knowledge, this is the first hardness result in the correlation-robust framework.

Given this main negative result, in the last part of our paper we ask whether certain standard assumptions on distributions, when satisfied by the marginals, can mitigate this impossibility, and possibly mitigate also some other undesirable phenomena that we identify. Our third and final main result gives a positive answer to this question, showing that if all of the marginal distributions exhibit \emph{monotone hazard rate (MHR)}, a standard tail condition in mechanism design, then the worst-case optimal revenue (from any mechanism) can be approximated up to a factor of $3.5$ using a simple pricing with many desirable properties; it concisely depends on the marginals only through one single-dimensional statistic of each, its revenue is agnostic of the unknown correlations, and its revenue is monotone in the given marginals:

\begin{theorem}[See \cref{mhr}]
Given $n$ MHR (marginal) distributions $F_1,\ldots,F_n$ from which a unit-demand buyer's valuations of $n$ respective items are drawn,
let $i$ be such that the median $p_i$ of $F_i$ is (weakly) higher than the medians of all other $F_j$. 
Setting a price of $p_i$ for item $i$ and setting a price of infinity for every other item maximizes up to a factor of at most $3.5$ (across all mechanisms whatsoever) the guaranteed (expected) revenue from the buyer over any correlation among the $n$ given distributions.
\end{theorem}

Unfortunately, if the MHR condition is relaxed,
we show that many hardships arise already for regular distributions (the ``one notch weaker'' standard tail condition in mechanism design), such as requiring linearly many different item prices (and hence requiring nonmonotone mechanisms), and nontrivial dependence on the marginals. This motivates the main question that we leave open:

\begin{open}
Is there a computationally efficient algorithm that given $n$ regular distributions $F_1,\ldots,F_n$ from which a unit-demand buyer's valuations of $n$ respective items are drawn, finds prices $p_1,\ldots,p_n$ that maximize up to a constant factor (even only across all possible choices of prices) the guaranteed (expected) revenue from the buyer over any correlation among the $n$ given distributions?
\end{open}

We conclude this paper by presenting some observations and examples, which may be useful toward this open question.

\subsection{Comparison with Bayesian Optimal Mechanisms under Independence}

The notion of correlation-robust mechanism design stands in contrast to the literature that assumes that buyer values are independent across items.
As such, it is worthwhile to draw comparisons across the two settings. 

Given the positive results of Carroll~\cite{Carroll17} and of Gravin and Lu~\cite{GravinL18}, one may have indeed wondered,
for a given multi-item setting (or at least for a canonical simple multi-item setting) and given marginals, whether correlation-robust mechanism design is always ``easier'' in some sense than Bayesian mechanism design with the assumption of independence of the marginals. Our results give a negative answer to this hope, in a precise formal sense.
Consider the best approximation factor to the optimal revenue (of any mechanism) that is obtainable by an efficiently computable pricing mechanism. In the Bayesian independent model, for both the additive and the unit-demand settings (two special cases of gross substitutes valuations), this factor is a constant that is strictly greater (worse) than $1$~\cite{ChawlaHMS10,ChawlaHK07,BabaioffILW14,Thanassoulis04}, and this extends 
even to subadditive settings~\cite{RubinsteinW18}, a strict superclass of gross substitutes.
In the correlation-robust model, for the additive setting this factor is simply~$1$~\cite{Carroll17} (i.e., optimal---an improvement compared to the Bayesian independent model). 
Yet, we show, in sharp contrast, that already for the unit-demand setting 
this factor is not only worse than in the independent model, but is actually unboundedly large as $n$ grows large, even when compared to the weaker benchmark of the optimal revenue of a pricing mechanism.

In this vein, it is also instructive to consider a recent paper by Bei et al.~\cite{BeiGLT19}. In their paper, they study correlation-robust pricing in a single-item setting with \emph{multiple} buyers, and give a pricing mechanism that maximizes the worst-case revenue up to a constant.\footnote{To the best of our knowledge, Bei et al.\ are the first to study approximate revenue maximization in a correlation-robust setting.} Given the well-known connections, in the Bayesian setting, between the single-item multi-buyer case and the multi-item single-buyer-with-unit-demand case~\cite{ChawlaHK07},
one may be surprised by the stark contrast between the positive result of Bei et al.\ and our negative one. This contrast in fact highlights a key difference between their analysis and ours: the mechanism used by Bei at al.\ in the single-item multi-buyer setting offers the item to different buyers for different prices, and does so by offering it first to the buyer for whom the set price is highest, then to the buyer for whom the set price is the second-highest, etc. That is, in their (single-item multi-buyer) setting the mechanism can force the item to be sold for the \emph{highest price} that is below the corresponding value, while in our (multi-item single-buyer) setting, since there is a single buyer making the decision, we have no escape from the price to be paid being the one that is farthest below the corresponding value (i.e., generating the highest \emph{utility}), allowing lower-priced items to cannibalize, as discussed above, from the sale probability of higher-priced items.
Indeed, to the best of our knowledge, our study of the unit-demand case is the first in a correlation-robust model where there is no solution that can be expressed as a composition of solutions to single-item auction setting. Approaching our research questions therefore required the development of new technical approaches to the correlation-robust model.

Taken in concert,
the above observations highlight correlation-robust revenue maximization as a framework for which intuition from the Bayesian setting may fail, and
for which completely new separate intuition
may have to be developed.

\subsection{Other Related Work} 

Gravin and Lu \cite{GravinL18} give an alternative proof to Carroll's result, and furthermore extend their study to solve the more general scenario of additive valuations with a buyer budget (the buyer's fixed budget is known to the seller, but the correlation between valuation distributions is again adversarially chosen), for which, as noted above, they show that Carroll's main message of the optimal mechanism being simple (and easy to compute) and its revenue being agnostic of the correlation still holds. Driven by similar motivation of robust revenue maximization, however in a different setting of contract (rather than auction) design, Duetting et al.~\cite{DuttingRT19} study the worst-case optimal contract in a principal-agent setting where only the expected rewards from the agent's actions to the principal are known (rather than the full reward distributions). Their identification of linear contracts as optimal in this sense once again features the same properties of simplicity, computational efficiency, and agnosticism to the hidden information.\footnote{Earlier work of \cite{Carroll15} gives a different sense in which linear contracts are max-min optimal, one that does not share this property of being agnostic to the unknown component.}

These works on robust mechanism design fit within a broader research agenda of \emph{robust optimization}, which has been studied in Operations Research tracing back to the classic paper of Scarf \cite{Scarf58}. This approach has been applied to mechanism design, among other domains \cite{BandiB14}. Within computer science, robust design ties in to ``beyond worst-case'' analysis approaches \cite{Roughgarden19}, and in particular to semi-random models \cite{BlumS95}. 
From this perspective, we study a semi-random model in which one aspect of the model (the item values) is randomly
drawn, whereas another (the correlation among values) is adversarially chosen.
Hybrid models such as these are gaining traction in recent years, in part due to their power to explain why certain algorithmic methods work better than expected. 
 
Most of the work on mechanism design for a unit-demand buyer has been in the standard Bayesian model.
If there is only a single item for sale, Myerson \cite{Myerson} characterizes the revenue-optimal mechanism, which for a single buyer (see also \cite{RileyZeckhauser}) amounts to simply setting the \emph{monopoly price} for the item (the price that maximizes the expected revenue given the item's value distribution). 

For the case of a single unit-demand buyer with a product distribution over item values, 
\citet{ChenDPSY18} show that computing a revenue-optimal pricing is NP-hard, even for identical distributions with support size $3$ (but can be solved in polynomial time for distributions of support size $2$).
\citet{ChawlaHK07,ChawlaHMS10} give a constant factor approximation for the optimal pricing, which applies also with respect to the optimal randomized mechanism (i.e., pricing lotteries) by the observation that pricing lotteries cannot increase revenue by more than a factor of $4$ in the case of product distributions \cite{ChawlaMS15}.
\citet{CaiD11,CaiD15} give an additive PTAS for the case of bounded distributions, and also derive structural properties of the optimal solution for special cases.
Among other properties, they show that if the buyer's values are independently distributed according to MHR distributions, then constant approximation can be obtained with a single price (which can be efficiently computed). Moreover, if values are also identically distributed, then a single price yields near-optimal revenue. 

For the case of correlated distributions, \citet{BriestK11} show that the optimal pricing problem does not even admit a polynomial-time constant approximation. 
It has been also shown that, unlike product distributions, pricing lotteries over items can increase the revenue (beyond item pricing) by a factor of $\log(n)$ \cite{Thanassoulis04,BriestCKW10}. 

Finally, another popular fairly recent line of research that builds upon the Bayesian setting by having the seller only have partial information about the underlying distribution (but keeping the optimal auction for the underlying distribution as the benchmark) is that of revenue maximization from samples, where the seller is shown samples from the underlying distribution rather than the whole distribution~\cite{ColeRoughgarden14}. Within this context, for pricing for a unit-demand buyer see \cite{MorgensternR16}, and for revenue maximization for a unit-demand buyer see \cite{GonczarowskiW18}.

\section{Preliminaries}
\label{sec:prelim}
\subsection{Model}
\label{sub:prelim-model}

The main player in our model is a \emph{seller}, who has $n\geq 2$ items for sale to a single unit-demand buyer. The buyer has a \emph{valuation profile} $v=(v_1,\dots,v_n)$ where $v_j\geq 0$ denotes her value for item~$j$. The seller can set a \emph{pricing} $p$, i.e., a vector of item prices $(p_1,\dots,p_n)$. 
Given a pricing $p$, the buyer purchases a single item $j$ that maximizes her (quasi-linear) utility $v_j-p_j$,\footnote{Tie-breaking is discussed below.} 
or nothing if her utility from purchasing any item would be negative.

\paragraph{Problem instance.}
An instance of our model consists of $n$ marginal distributions $F_1,\dots,F_n$ for the $n$ item values $v_1,\dots,v_n$.
The distributions are over supports $V_1,\dots,V_n$ known as \emph{value spaces}.
If~$V_j$ is bounded, then we denote its maximum value by $v_j^{\max}$, and we denote its minimum value by~$v_j^{\min}$.  
Importantly, item values can be \emph{correlated}, as long as for every item $j$ the buyer's value $v_j$ is \emph{marginally} distributed according to $F_j$.\footnote{Being marginally distributed according to $F_j$ means the following: let $\Pr_{v\sim F}[v_j\leq x]$ be the probability that if we sample a valuation profile $v$, the value for item $j$ is at most $x$; then $F_j(x)=\Pr_{v\sim F}[v_j\leq x]$ for every $x\in V_j$.}	
Notice that the correlation is \emph{not} part of the instance but rather will be adversarially chosen. 

Given a value $v_j$, its quantile $q_j(v_j)$ is $F_j(v_j)=\Pr_{\nu \sim F_j}[\nu\le v_j]$ (so low quantiles correspond to low values and vice versa).
For every quantile $q_j\in[0,1]$ we define its value $v_j(q_j)$ to be $\min\{v\mid F_j(v)\ge q_j\}$ (notice that if $F_j$ is strictly increasing then it is \emph{invertible}, the inverse function $F_j^{-1}(\cdot)$ is well defined, and $v_j(q_j)$ is equal to the inverse $F_j^{-1}(q_j)$).

\paragraph{Compatible distributions.}
For a given instance with marginals $F_1,\dots,F_n$, a \emph{compatible distribution} $F$ is a joint distribution over valuation profiles $v$ such that the  marginals of $F$ for the individual item values coincide with $F_1,\dots,F_n$. 
A natural class of compatible distributions is that of \emph{perfect couplings}. 
In a perfect coupling, the value of any one of the items determines the values of all others. 
More formally, a distribution $F$ is a perfect coupling if for every item $j\in[n]$ there exists a \emph{coupling function}, a measure-preserving bijection $C_j:[0,1]\to[0,1]$, and a valuation profile is drawn from $F$ by randomly drawing $q\sim U[0,1]$, and taking the value of each item $j$ to be $v_j(q_j)$ where $q_j=C_j(q)$. 

One particular perfect coupling that plays a role in our results is the following. 

\begin{definition}
	The \emph{comonotonic distribution} $\com$ is the perfect coupling defined by $C_j(q)=q$ for every $j\in[n]$.
\end{definition}

The comonotonic distribution appears in the work of \cite{Carroll17} on correlation-robust pricing for an \emph{additive} buyer.\footnote{In that setting, as is shown there, when the marginals are regular distributions, the worst correlation from the seller's perspective (in a sense formalized below) is the comonotonic distribution.} Intuitively, this distribution is the compatible distribution in which values are ``as positively correlated as possible.'' For example, observe that if all the marginals are identical, the in every valuation profile drawn from $F$ all the values are identical. 

\subsection{The Max-Min Pricing Problem}
\label{sub:prelim-problem-def}

\paragraph{Objective.} 

For a given instance with marginals $F_1,\dots, F_n$, denote by $R(p,F)$ the seller's expected revenue from setting a pricing $p$ if the valuation profile $v$ is drawn from a compatible distribution~$F$: 
\begin{equation*}
R(p,F)=\mathbb{E}_{v\sim F}[p_{j^*(v,p)}],
\end{equation*}
where $j^*(v,p)$ is the item purchased by the buyer.
We note that the buyer has to break ties between items that give her the same utility.
We assume that tie-breaking between any two items depends only on the identities of the two items that give the same utility, the value of that utility, and the prices of the two items. The tie-breaking rule must also be consistent (no cycles). For example, breaking ties in favor of higher-priced items and then by index, or in favor of lower-priced items and then by index, are both allowed.

In the \emph{max-min pricing problem}, the goal is to find a pricing $p$ that maximizes the minimum, over all compatible distributions $F$, of the expected revenue $R(p, F)$. 

We introduce the following notation to make this formal: Let $R(p)$ be the \emph{robust revenue guarantee} of $p$, i.e., the worst expected revenue from $p$ over all compatible distributions:
\begin{equation*}
R(p)=\inf_{\text{compatible }F}R(p,F),
\end{equation*}
and let $R^*$ be the optimal robust revenue guarantee over all pricings:
\begin{equation*}
R^*=\sup_{p}R(p)=\sup_{p}\inf_{\text{compatible }F}R(p,F).
\end{equation*}

A pricing $p$ is \emph{$\alpha$-max-min optimal}, for $\alpha\ge 1$, if $R(p)\ge \frac{1}{\alpha}R^*$; if $\alpha=1$, then we say that $p$ is a max-min optimal pricing.
The max-min pricing problem is to find, for a given problem instance, an $\alpha$-max-min optimal pricing $p$ with $\alpha$ as close as possible to 1.

\paragraph{Zero-sum game and the Adversary's perspective.}

In the max-min pricing problem, the seller can be viewed as a player in a zero-sum game corresponding to the problem instance, in which the \emph{Adversary}'s strategy space is the space of all compatible distributions. The seller's payoff for ``playing'' a pricing $p$ against a compatible distribution $F$ is $R(p,F)$. 
The Adversary's goal is to choose $F$ that minimizes $R(p,F)$. We refer to a distribution achieving $\inf_{\text{compatible }F} R(p,F)$ (if it exists) as a \emph{best response of the Adversary} to the pricing $p$. (More generally one can consider a $\beta$-best response, which is a distribution $F$ achieving $R(p,F)\le \beta\cdot\inf_{\text{compatible }F} R(p,F)$.)
 
\paragraph{Simple solution classes.}

Arguably the two simplest possible pricing classes are the following.

\begin{definition}
	A pricing $p$ is a \emph{single price} if all prices but one are $\infty$.
\end{definition}

\begin{definition}
	A pricing $p$ is \emph{uniform} if all prices are equal.
\end{definition}

A single price $p$ has a particularly nice robustness property: it is \emph{correlation agnostic}. That is, its expected revenue is the same against \emph{any} compatible distribution. I.e., $R(p,F)=R(p,F')$ for every compatible $F,F'$. Such robustness is a recurring theme in the literature on robust mechanism design \cite{Carroll19}; in particular it makes the task of showing that $p$ is ($\alpha$-)max-min optimal much simpler.\footnote{E.g., it is sufficient to show a compatible distribution for which $p$ is the seller's $\alpha$-best response, similar to \cite{Carroll17}, or to show for every other pricing $p'$ a compatible distribution $F'$ such that $R(p,F')\ge \alpha R(p',F')$, similar to \cite{DuttingRT19}.\label{ftn:sufficient}} In comparison, uniform pricings do not enjoy the robustness property, however they ``dominate'' single prices in the following sense: one can naturally turn any single price $p$ into a uniform pricing $p'$ (by setting the prices of all items to be the same as that single price) such that $R(p,F)\le R(p',F)$ for every $F$.

\section{The Adversary's Best Response, and\texorpdfstring{\\}{ }Robust Revenue Guarantee Calculation}
\label{sec:adv-BR}
In this section, we will first present an algorithm that finds the exact  best response of the Adversary for any given pricing (and along with the best response, also the robust revenue guarantee of this pricing) when every marginal is a uniform distribution over a finite multi-sets of values, and all such multi-sets are of the same size. This algorithm runs in polynomial time in the size of the input (for explicitly given such multi-sets).
We will then show that for arbitrary finite distributions (even if the probabilities are not even rational numbers), we can still output the best response of the Adversary (and the robust revenue guarantee) using a slightly generalized version of this algorithm, and do so in polynomial time. Finally, we explain how to further use our algorithm to get arbitrarily close to the best response of the Adversary in the general case of arbitrary (not necessarily discrete) marginal distributions, where as we will show a precise best response may not exist.
Proofs omitted from this section appear in \cref{appx:proofs-alg}. 

\subsection{Perfect Couplings of Uniform Distributions over Multisets of Identical Sizes}\label{uniform-over-multisets}
In this section we assume that every marginal is a uniform distribution over a finite multi-set of values, where the multi-sets corresponding to the various marginals are all of the same size $d$.
(In \cref{sec:adv-BR-general} we will show that the algorithm that we will present in the current section can be tweaked to to work for any discrete distribution,
and still run in polynomial time.)

Given a pricing, we wish to find the worst distribution for those prices that is compatible with the marginals. We will show that there is a perfect coupling that minimizes the revenue over all compatible distributions.
 
To handle the possibility of the buyer not buying any item, we assume, without loss of generality, that one of the items is a special ``null item'' that always has value $0$ and will be priced at $0$, such that the buyer buying this null item corresponds to the buyer not buying any item. Other than that we will treat the null item as any other item (assume it has price 0 and that its value is distributed uniformly over a multi-set of all-$0$ values, and thus the corresponding utilities are also all $0$). 

Assume without loss of generality that items are ordered such that prices are nondecreasing, $p_1 \leq p_2\leq \cdots \leq p_n$.
Let $v_i$ be the vector of values of item $i$ sorted in nonincreasing order; i.e., $v_i^1 \geq \cdots \geq v_i^d$ (recall that by assumption, the value of item $i$ is drawn uniformly from the $d$ values $v_i^1,\ldots,v_i^d$).
Given the prices we can transform this vector to a vector of utilities that all have the same probability (each has probability $\nicefrac{1}{d}$) with $u_i^j=v_i^j-p_i$ for every item $i$ and index~$j\in \{1,\ldots,d\}$.
We thus obtain a vector 
$u_i$ of utilities of items  $i$ sorted in nonincreasing order; i.e., $u_i^1 \geq u_i^2 \geq \cdots \geq u_i^d$.

We note that the Adversary's best response may depend on the tie-breaking rule used by the buyer to choose among items that yield the same utility.
Our algorithm for the Adversary's best response will break ties in the same way as the buyer does. 
We say that $u_i^t$ is dominated by $u_j^r$, and denote this by $u_i^t \prec u_j^r$, if either $u_i^t < u_j^r$, or $u_i^t = u_j^r$ and the buyer when facing the choice between buying item $i$ and buying item $j$, either at utility $u=u_i^t = u_j^r$, breaks this tie in favor of item $j$. 

A \emph{perfect coupling} (or simply a \emph{coupling}) in this setting corresponds to a bijection for each $i$, from indices $\{1,2,\ldots,d\}$ to the multi-set of utilities (equivalently, to the the multi-set of values) of item $i$,
where to draw the utilities for the $n$ items, an index is drawn uniformly at random, and then the utility for each item is determined according to the bijection of that item, applied to this index.
We will think about the image of each index under all $n$ bijections as a \emph{chain}, coupling together $n$ elements in the multi-sets of utilities, one element for each item. (So to draw the utilities for the items, one of these chains is simply drawn uniformly at random.)
Hence, a (perfect) coupling can be described using $d$ chains that form a partition of all utilities, with each chain containing exactly one utility for every item (and every such $d$ chains describe a perfect coupling).
Given a chain $t$ in the coupling, we denote the utility of item $i$ in $t$ by $u_i(t)$.
A chain $t$ in which $u_j(t)\prec u_i(t)$ for all $j\neq i$ is said to be dominated by item $i$---this is the item that will be bought if chain $t$ is drawn. The expected revenue from a given perfect coupling is simply the weighted average of the prices of the items, each weighted proportionally to the number of chains in this coupling that it dominates.

\subsubsection{The Adversary's Algorithm}

The algorithm for the Adversary's best response is given in \cref{alg:adv-BR}. It receives as input item prices $p_1 \leq p_2 \leq \cdots \leq p_n$ and utility vectors  $u_1, \ldots, u_n$ for the $n$ items, with each $u_i$ sorted in nonincreasing order, and returns a coupling $c$. When the algorithm is run, $c$ is initially empty, and is augmented by chains (and has its chains modified at times) throughout the course of the algorithm. We abuse notation and write $u_i^k \in c$ if $u_i^k \in t$ for some chain $t \in c$ (and write~$u_i^k \not\in c$ otherwise).
\begin{algorithm}[ht]
	\caption{The Adversary's best-response algorithm; Input: utilities $u_1,\ldots,u_n$; prices $p_1 \leq \cdots \leq p_n$.}
	\label{alg:adv-BR}
	\begin{algorithmic}[1]
		\STATE $c \gets \emptyset$ 
		\FOR{$i=1, \ldots, n$}
		\FOR{$k = 1, \ldots, d$}\label[line]{maximize-begin}
		\IF{for every item $j>i$ there exists $u_j^{\ell} \notin c$ s.t. $u_j^{\ell} \prec u_i^k$}
		\STATE For every $j>i$, let $\ell_j =  \arg\max_{\ell}\{u_j^{\ell} \mid u_j^{\ell}\not\in c$ and $u_j^{\ell}\prec u_i^k\}$ 
		\COMMENT{break ties towards a lower index}
		\STATE For every $j<i$, let $\ell_j$ be an arbitrary index s.t. $u_j^{\ell_j} \not\in c$
		\STATE $c \gets \bigl(u_i^k,\{u_j^{\ell_j}\}_{j\neq i}\bigr)$ \label[line]{adv-BR-high}
		\ENDIF
		\ENDFOR\label[line]{maximize-end}
		\STATE Let $m=|c|$ and let $t_1, \ldots, t_m$ be the chains in $c$ sorted such that $u_{i+1}(t_1) \leq \cdots \leq u_{i+1}(t_m)$ 
		\FOR{$r=1,\ldots,m$}\label[line]{recouple-begin}
		\STATE $t_r \gets t_r \setminus \bigl\{u_{i+1}(t_r)\bigr\} \cup \{u_{i+1}^{d-r+1}\}$ \COMMENT{i.e., in chain $t_r$, replace $u_{i+1}(t_r)$ with the smallest $u_{i+1}^\ell\notin\cup_{j< r}t_j$}\label[line]{adv-BR-low}
		\ENDFOR\label[line]{recouple-end}
		\ENDFOR
		\RETURN $c$
	\end{algorithmic}
\end{algorithm}
The algorithm first attempts to build as many chains as possible that are dominated by item $1$ (Lines \labelcref{maximize-begin}--\labelcref{maximize-end} when $i=1$), then as many chains as possible that are dominated by item $2$ (Lines \labelcref{maximize-begin}--\labelcref{maximize-end} when $i=2$), etc. When building a chain dominated by a certain utility for item $1$, the algorithm attempts to use the highest possible utility for each higher-priced item $j$ that would still be dominated by that utility for item~$1$, in order to leave lower utilities for item $j$ to possibly be dominated in future chains by lower utilities for item $1$ or by utilities for some $i<j$. Just before turning to build chains dominated by this item $j$, though, the algorithm has a transition stage (Lines \labelcref{recouple-begin}--\labelcref{recouple-end}) that recouples all of the chains built so far to use the lowest, rather than highest, utilities for item $j$, since from that moment onward, a high utility for item $j$ is no longer a liability that we attempt to dominate by utilities for lower-priced items (to have those items cannibalize from item $j$), but rather an asset with which to attempt to dominate utilities for even higher-priced items (to have item $j$ cannibalize from those items).
\cref{fig:alg-illustration} illustrates an execution of the algorithm.
\begin{figure}[ht]
	\centering
	\includegraphics[scale=0.49]{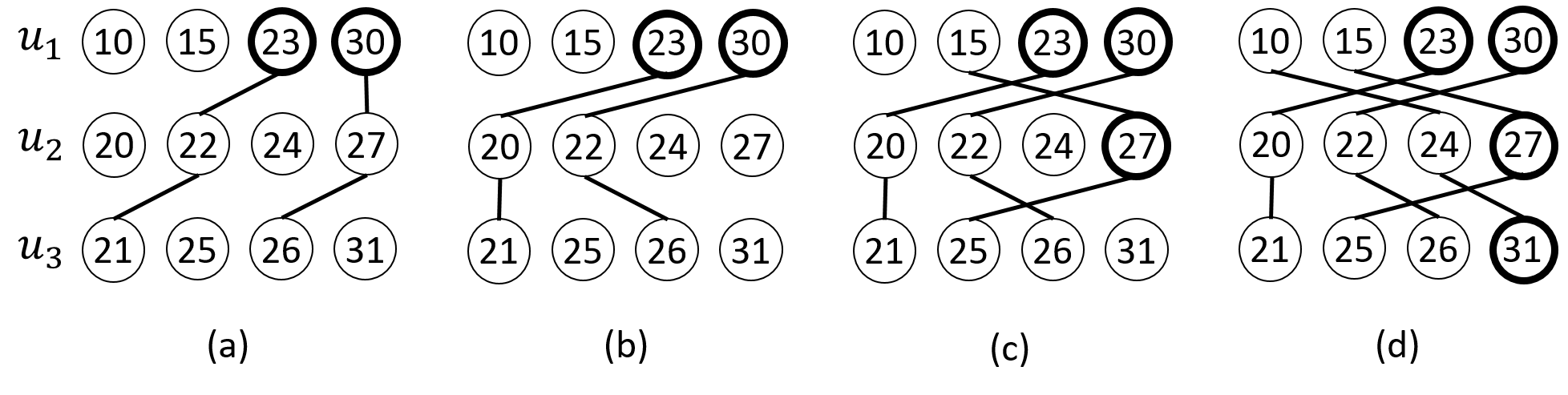}\vspace{-1em}
	\caption{An illustration of \cref{alg:adv-BR}. (a) Following Iteration $i=1$, just before the transition stage (i.e., following \cref{maximize-end}); (b) following the transition between Iteration $i=1$ and Iteration $i=2$ (i.e., following \cref{recouple-end}); (c) following Iteration $i=2$; (d) following Iteration $i=3$. Thick cycles signify the corresponding dominant utility in every chain. Note that without the transition stage between Iterations $i=1$ and $i=2$, we would have ended up with two chains dominated by item 1 and two chains dominated by item 3, and hence would have failed to minimizing the revenue.}
	\label{fig:alg-illustration}
\end{figure}

\subsubsection{Optimality of Algorithm~\ref{alg:adv-BR}}

\cref{alg:adv-BR} returns a (perfect) coupling that defines a distribution $F$ that is compatible with the marginals.  
{In this section we will show that no other perfect coupling defines a distribution that generates worse revenue. (In \cref{sec:adv-BR-general} we will show that no other distribution that is compatible with the marginals, whether or not induced by a perfect coupling, can generate any worse revenue.) Before we formally state this claim as \cref{cor:alg}, let us introduce some notation.}

Let $K_i$ be the maximum number of chains dominated by item $i$ in any coupling.
Let $K_{[i]}$ be the maximum number of chains dominated by one of items $1, \ldots, i$ in any coupling.
Given a coupling $c$, let $k_i(c)$ (respectively, $k_{[i]}(c)$) be the number of chains in $c$ dominated by item $i$ (resp., by one of items $1, \ldots, i$).
A coupling $c$ such that $k_{[i]}(c) = K_{[i]}$ is said to \emph{realize} $K_{[i]}$.
As standard, we denote the set $\{1, \ldots, n\}$ by $[n]$.
Our main result for this section is:

\begin{theorem}
	\label{cor:alg}
	The coupling output by \cref{alg:adv-BR} simultaneously realizes $K_{[i]}$ for every $i \in [n]$.
\end{theorem}

The proof of \cref{cor:alg} relies on \cref{lem:residual-optimality,prop:sim-max} stated below.

\begin{lemma}
	\label{lem:residual-optimality}
	The coupling $c$ output by \cref{alg:adv-BR} satisfies: (1) $k_1(c)=K_1$, (2) for every $i \geq 2$, $k_i(c) = \max\bigl\{k_i(c') ~\big|~$c'$~\text{is a coupling}\And k_j(c')=k_j(c) \quad \forall j<i\bigr\}$; i.e., $c$ maximizes the number of chains dominated by item $i$ fixing the number of chains dominated by items $j<i$. 
\end{lemma}

\begin{proposition}
	\label{prop:sim-max}
	There exists a coupling that simultaneously realizes $K_{[i]}$ for every $i \in [n]$.
\end{proposition} 

The combination of \cref{prop:sim-max,lem:residual-optimality} implies \cref{cor:alg}, showing that no coupling defines a distribution that generates worse revenue than that defined by the coupling output by \cref{alg:adv-BR}. We now establish the proof of \cref{prop:sim-max}. The remainder of the proofs are relegated to \cref{appx:proofs-alg}.

\begin{proof}[Proof of \cref{prop:sim-max}]
We first prove that there exists a coupling that simultaneously realizes $K_{[1]}$ and $K_{[2]}$, and then show how to generalize this to any prefix. 
Suppose toward contradiction that for every coupling $c$ that realizes $K_{[2]}$ it is the case that $k_1(c) < K_1$. 
Let $k'_1 = \max\bigl\{k_1(c) ~\big|~ c \mbox{ realizes } K_{[2]}\bigr\}$.

Let $u_{1,2}$ be the joint vector of utilities of items 1 and 2, sorted in nondecreasing order ($u_{1,2}$ is of length $2d$).
We use the term $k$ \emph{top utilities} of a utility vector to refer the $k$ highest utilities in the vector, breaking ties in favor of lower indexes. 

Let $S$ be the set of the top $k'_1$ utilities of item 1 and the top $K_{[2]}-k'_1$ utilities of item 2 in $u_{1,2}$. 
Sort the utilities in $S$ in nonincreasing order and couple them (in this order) with the highest-possible utilities of items $3, \ldots, n$ into disjoint chains.
This gives $K_{[2]}$ disjoint chains, partitioned into those \emph{rooted at item 1} and those \emph{rooted at item 2}.
(Note that every chain rooted at item 1 can be augmented with a suitable utility of item 2 such that the obtained chain is dominated by item 1, and analogously for chains rooted at item 2, since at most $d$ utilities of $u_{1,2}$ are coupled and all other utilities are smaller than the corresponding coupled ones.)

We will prove that the lowest chain rooted at item 2 can be replaced by a new chain rooted at utility $u_1^{k'_1+1}$, still realizing $K_{[2]}$, thus
contradicting the maximality of $k'_1$. 

\begin{figure}
	\centering
	\includegraphics[scale=0.54]{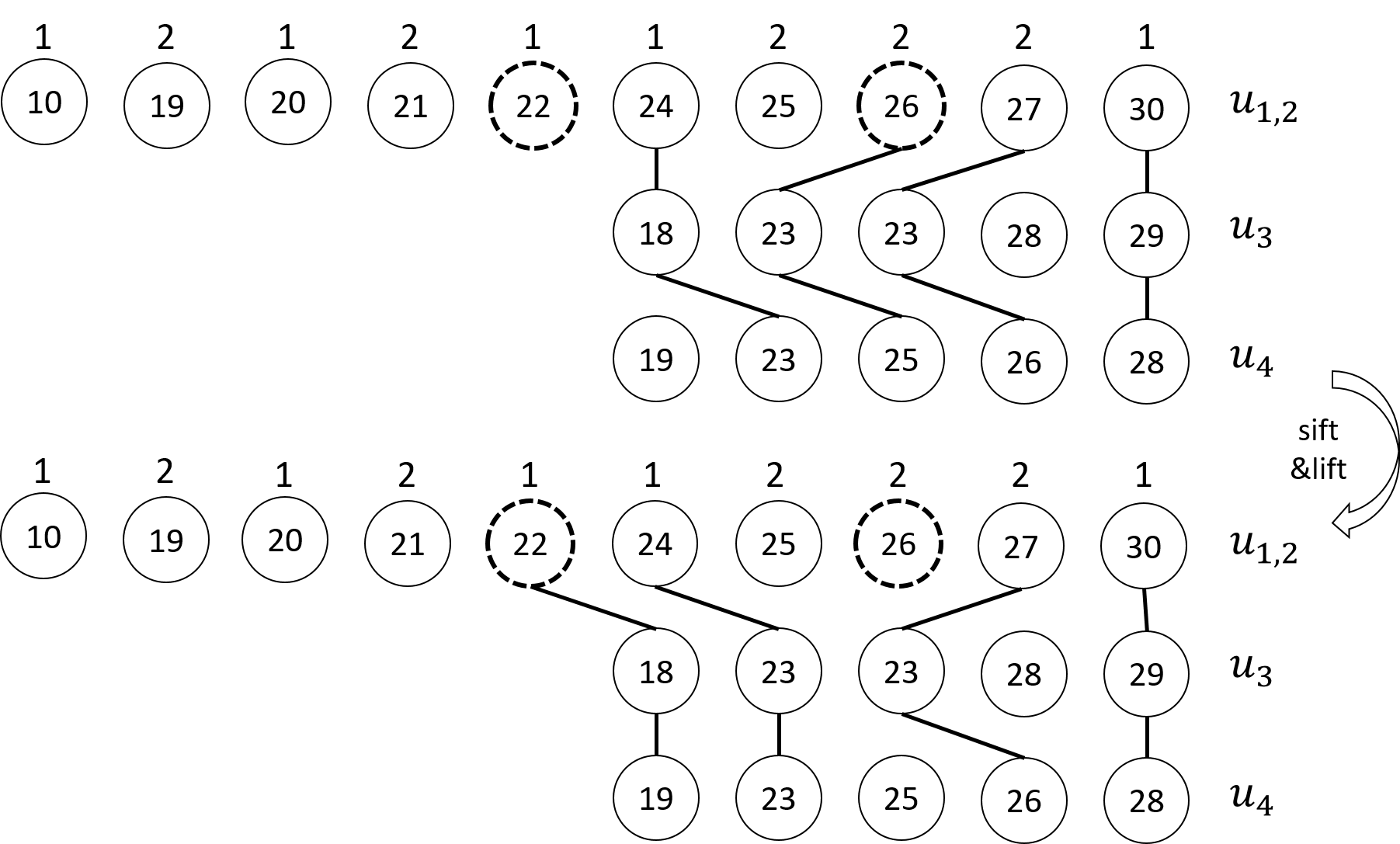}
	\caption{An illustration of the sift\&lift process. $d=5$. In the top coupling $c$, we have $k_1(c)=2$ and $k_2(c)=2$ (hence $k_{[2]}(c)=4$). In the bottom coupling $c'$ (following the sift\&lift process, and the consequent addition of a chain dominated by item 1), $k_1(c')=3$ and $k_2(c')=1$ (hence still $k_{[2]}(c')=4$). Dashed circles signify utilities whose status changed: a chain rooted at a utility of item 2 (and hence dominated by item 2) was removed, enabling (after the lifting process) the creation of an additional chain rooted at a utility of item 1 (and hence dominated by item 1).}
	\label{fig:remove-n-shift}
\end{figure}

We will now define a process that we call \emph{sift\&lift}, which removes the chain rooted at the lowest utility of item 2 ($u_2^{K_{[2]}-k'_1}$), and ``lifts" all subsequent chains rooted at item 1 (i.e., whenever possible, uses newly available higher utilities of items $3, \ldots, n$). 
The sift\&lift process is formally specified in \cref{alg:sift-lift} in \cref{appx:sift-lift}, and an illustration is given in \cref{fig:remove-n-shift}.
The following lemma, whose proof we relegate to \cref{appx:proofs-alg}, shows that after the sift\&lift process, one more chain dominated by item 1 can be added to the coupling. 

\begin{lemma}
	After removing the chain rooted at $u_2^{K_{[2]}-k'_1}$ and lifting the subsequent chains rooted at item 1, for every $i \ge 3$ it is the case that $u_i^d$ is uncoupled and $u_i^d \prec u_1^{k'_1+1}$.
\label{lem:add-1}
\end{lemma}

\cref{lem:add-1} implies that one more chain rooted at item 1 can be added, at the expense of the bottom chain rooted at item 2. Thus, the resulting coupling contains $k'_1+1$ chains dominated by item 1, and $K_{[2]}$ chains dominated by one of items 1,2, contradicting the maximality of $k'_1$. 
We conclude that there exists a coupling that simultaneously realizes $K_{[1]}$ and $K_{[2]}$, as desired.

We next show how to extend the proof to any $i \in [n]$. 
Suppose toward contradiction this is false. 
Let $\ell \in [n]$ be the smallest value for which this is false, and let $x \in [\ell]$ be the smallest value such that $K_{[\ell]}$ cannot be simultaneously realized with $K_{[1]}, \ldots, K_{[x]}$.

Repeat the same process as in the case of items 1,2, with items $1,\ldots,x$ in the role of item 1, and items $x+1,\ldots,\ell$ in the role of items 2.
By the same argument as before, one can increase the number of chains dominated by one of items $1,\ldots,x$ without affecting any $k_{[j]}$ for $j \in [\ell]$, contradicting the definition of $x$.
\end{proof}

\subsection{General Discrete Distributions, Computational Efficiency,\texorpdfstring{\\}{ }and Further Extensions}
\label{sec:adv-BR-general}

Recall that there are three gaps between \cref{cor:alg,thm:alg-is-opt}, which we will prove in this appendix: the first gap is that \cref{alg:adv-BR}, as stated, works only for marginals that are uniform over multisets (implying applicability also for discrete marginals whose probabilities for the various values are rational numbers, but not irrational ones), the second gap is that it only guarantees that the coupling that we find generates the worst revenue of any coupling rather the worst revenue of any correlated distribution,\footnote{This is possibly the most subtle of the three gaps. To better understand it, note that for $n=2$ items, the Birkhoff--von~Neumann theorem tells us that any distribution over pairs of utilities of the two items is a convex combination of distributions defined by a perfect couplings, and so there is a perfect coupling that generates at most as much revenue. For $n>2$ items, though, it is well known that the Birkhoff--von Neumann theorem fails to hold \cite{LinialL2014}, and so there are distributions over $n$-tuples of utilities of the $n$ items that cannot be generated by a two-step process of first drawing a perfect coupling and then drawing a chain of values from that coupling. We will therefore have to derive our result not using such a general tool, but as property of the specific problem that we are considering.} and the third gap is that the computational efficiency of the presented algorithm is polynomial in $d$, so when applied to discrete distributions with arbitrary rational probabilities is actually polynomial in the lowest common denominator of all of these probabilities.

In this section, we show how to modify \cref{alg:adv-BR} to bridge all three of these gaps in one fell swoop, by reinterpreting \cref{alg:adv-BR} as a water-filling algorithm that arbitrarily discretizes (for ease of presentation and analysis, but for no real constraint) its step size to $\nicefrac{1}{d}$.
The idea underlying our modification of \cref{alg:adv-BR} to bridge these three gaps is quite simple: instead of splitting every probability mass \emph{in advance} of running the algorithm, we will split probability masses \emph{on demand}. Formally, we will allow the algorithm to split each probability mass into several probability mass ``nodes,'' however unlike in \cref{uniform-over-multisets}, these nodes may have different masses. The algorithm will output a coupling of such nodes, where $n$ nodes may be coupled only if they all have the same mass. This will allow us to handle irrational probabilities, to compare to arbitrary correlated distribution (as any correlated distribution can be represented as a perfect coupling of such nodes of unequal masses), and will additionally maintain computational efficiency that is polynomial in the size of the support of the given marginals. The pseudo-code of the algorithm is in fact virtually unchanged from that of \cref{alg:adv-BR}, however to interpret it for this scenario we introduce a few local semantic changes to our interpretation of this pseudo-code:
\begin{enumerate}
	\item
	We will continue to denote the vector of utilities from item price $p_i$ by $u_i$, however we will now have the utilities in each $u_i$ be distinct, and we will say that $u_i^\ell\in c$ only if the entire mass of $u_i^\ell$ is already coupled (that is, if this utility was split into smaller nodes, then all such nodes should be coupled for this to hold). The notation $u_i^\ell\notin c$ throughout \cref{alg:adv-BR} should therefore be interpreted as saying that some mass of $u_i^\ell$ is still uncoupled.
	\item
	In \cref{adv-BR-high} of \cref{alg:adv-BR}, when we write $c \gets (u_i^k,\{u_j^{\ell_j}\}_{j\neq i})$, we mean the following: let $p$ be the lowest remaining uncoupled mass of any of these utilities; split a new node of mass $p$ from each of these utilities, couple all of these $n$ nodes together, and add the resulting chain to $c$. (The reinterpreted algorithm can therefore be thought of as a water-filling algorithm of sorts.) We note that each execution of \cref{adv-BR-high} causes at least one utility $u_j^\ell$ of one item in $1,\ldots,n$ to change from satisfying $u_j^\ell\notin c$ to satisfying $u_j^\ell\in c$, so the computational complexity does not explode.
	\item
	In \cref{adv-BR-low} of \cref{alg:adv-BR}, the line captioned ``i.e., in chain $t_r$, replace $u_{i+1}(t_r)$ with the smallest $u_{i+1}^\ell\notin\cup_{j< r}t_j$,'' we will be perform the following, which still implements the exact same comment: let $p_r$ be the mass of all nodes in $t_r$. First, remove $u_{i+1}(t_r)$ from $t_r$, making that $p_r$ mass of $u_{i+1}(t_r)$ uncoupled again (merging it back with any other uncoupled mass of $u_{i+1}(t_r)$). Then, take the $p_r$ lowest-utility uncoupled mass of item $i+1$ and add it to $t_r$ instead of the mass that we just removed from that chain (as in \cref{uniform-over-multisets}, there may be some mass that we removed and then added back, and this is fine). If the $p_r$ lowest-utility uncoupled mass of item $i+1$, which we wish to add to $t_r$, happens to span more than one utility---say that it spans $d'$ utilities---then we split the chain $t_r$ into $d'$ chains with appropriate masses such that each of these utilities of $i+1$ could be coupled with a different one of these chains. (Once again, this is consistent with thinking of the reinterpreted algorithm as a water-filling algorithm of sorts.) If this results in any chain that is over exact same utilities as an already existing chain, then we merge such chains.
\end{enumerate}

To analyze the newly-interpreted \cref{alg:adv-BR}, we redefine $K_i$ to be the maximum probability of item $i$ being sold in any correlated distribution compatible with the marginals, and $K_{[i]}$ to be the maximum probability of one of items $1, \ldots, i$ being sold in any correlated distribution compatible with the marginals.
Given a correlated distribution $F$, let $k_i(F)$ (respectively, $k_{[i]}(F)$) be the probability of item $i$ (resp., one of items $1, \ldots, i$) being sold in $F$ .
A correlated distribution $F$ such that $k_{[i]}(F) = K_{[i]}$ will be said to \emph{realize} $K_{[i]}$. Then, completely analogous arguments to those in the proof of \cref{cor:alg} give that the newly-interpreted \cref{alg:adv-BR} simultaneously realizes $K_{[i]}$ for every $i \in [n]$, which we equivalently restate as the main result of this \lcnamecref{sec:adv-BR}:

\begin{theorem}
	\label{thm:alg-is-opt}
	Given any prices $p_1 \leq p_2 \leq \cdots \leq p_n$, and discrete marginals $F_1,\ldots,F_n$, 
	the correlated distribution $F$ generated by the \cref{alg:adv-BR} with its semantics interpreted as in this section,
	attains the lowest expected revenue over all correlated distributions that are compatible with these marginals. Furthermore, if each marginal $F_i$ is given explicitly as a list of value-probability pairs $\bigl((v^\ell_i,p^\ell_i)\bigr)_{\ell}$, 
	then this algorithm runs in time polynomial in its input size.
\end{theorem}

Interestingly, since for the case of marginals that are uniform over multi-sets of $d$ values both interpretations of \cref{alg:adv-BR}---from \cref{uniform-over-multisets} and from this section---coincide, \cref{thm:alg-is-opt} implies that for such marginals even the interpretation from \cref{uniform-over-multisets} in fact generates the worst revenue among all correlated distributions compatible with the marginals, and not only among those defined by perfect couplings.

In \cref{appx:alg-extensions} we discuss extensions of the techniques of this section to nondiscrete distributions.

\section{Hardness of Approximation}
\label{sec:lower-bound}
In this section we will show a hardness of approximation for the max-min pricing problem.  Our hardness result will apply even when the support of every marginal is finite and the probability of sampling each value is a rational number.  Note that this allows the distributions to be provided explicitly as input to a max-min pricing algorithm, rather than through an oracle.  We show that even under this direct access model, it is NP-hard to compute prices that achieve an $o(n^{1/2-\epsilon})$-approximation, for any $\epsilon > 0$.  This is true regardless of the way in which ties are broken in case of buyer indifference.

\begin{theorem}
\label{thm:hardness}
For any $\epsilon > 0$, it is NP-hard to obtain an $O(n^{1/2 - \epsilon})$-approximation to the max-min pricing problem.
\end{theorem}

We will prove \cref{thm:hardness} by reducing from the maximum independent set (MIS) problem.  Recall that in the MIS problem, the input is an unweighted graph $G = (V,E)$ and the goal is to find an independent set $S \subseteq V$ of maximum size.\footnote{Set $S \subseteq V$ is an independent set if no two nodes in $S$ are adjacent.}  It is known to be NP-hard to achieve an $O(n^{1-\epsilon})$ approximation to the MIS problem, for any fixed $\epsilon > 0$~\cite{Zuckerman07}.

\subsection{Reduction Construction and Proof Outline}

Given an MIS instance $G = (V,E)$ on $n$ vertices, we will construct an instance of the max-min pricing problem as follows.  First, order the vertices of $V$ by labeling them $1, \dotsc, n$, arbitrarily.  We will define $n$ marginal distributions.  For each $i \in [n]$, distribution $F_i$ takes on value $2^{i n}$ with probability $2^{-i  n}$.  Moreover, for each vertex $j > i$ such that $\{i,j\} \in E$, distribution $F_i$ will take on value $2^{j  n}$ with probability $2^{-j  n} / \sqrt{n}$.  With all remaining probability $F_i$ will take on value $0$.\footnote{Note that we can restrict attention to graphs $G$ for which $n$ is a sufficiently large perfect square, so that all of these values and probabilities are rational numbers.}

The idea behind the above construction is as follows.  First, to get nonnegligible (as a function of $n$) revenue from any item $i$, even before taking into account any cannibalization,
the price of item $i$ must be set close to $2^{in}$.  Indeed,
such a price would yield revenue $1$ if that were the only item sold. A much lower price for that item (say, less than $2^{(i-1)n}$) would not increase the sale probability of that item and would therefore yield negligible revenue from it, while any higher price for that item (say, greater than $2^{in}$ but at most $2^{jn}$ for some $j>i$) would only sell item $i$ with probability around $2^{-jn}/\sqrt{n}$, resulting again in negligible revenue from it.  So let us say that an item $i$ priced between $2^{(i-1)n}$ and $2^{in}$ is \emph{reasonably priced}. 

If we reasonably price two items $i<j$ s.t.\ $\{i,j\}\in E$
then, under an appropriate correlation between $F_i$ and $F_j$ (which sets value $2^{jn}$ for item $i$ only when item $j$ has this value as well), item $i$ would cannibalize a $\nicefrac{1}{\sqrt{n}}$ fraction of the sale probability of item $j$. Taking this one step further, if we reasonably price an item $j$ and not one but $\sqrt{n}$ lower-index neighbors of $j$ then, under an appropriate correlation of the marginals, the revenue from item $j$ would be completely cannibalized.\footnote{So the factor of $\nicefrac{1}{\sqrt{n}}$ by which some probabilities in the above construction are multiplied was chosen to be small enough that the revenue from an ``unreasonably highly priced'' item would be negligible, but large enough that a few lower-priced neighbors of an item could together effectively cannibalize the revenue from that item.}  As it turns out, these are essentially the only meaningful cannibalizations possible under any correlation structure.

Reasonably pricing only items in an independent set, while pricing all other items at $+\infty$, would therefore obtain revenue at the order of the size of that independent set (see \cref{lem.IS.pricing.lb} below).  Thus, if $G$ has a large independent set, then a high revenue guarantee can be obtained. On the other hand, 
if $G$ only has very small independent sets, then reasonably pricing many items 
would inevitably mean that most items have at least $\sqrt{n}$ reasonably-priced neighbors, which makes it possible to cannibalize the revenue of most items (simultaneously) with an appropriate correlation structure (see \cref{lem:MIS.pricing.ub} below). 
This paves the way to differentiating between instances $G$ with large and small independent sets based on their optimal robust revenue guarantees.  In the remainder of this section we will make this argument precise.

\subsection{Proof}

To make the dependence on $G$ clear, we will write $F_i(G)$ for the $i$th marginal distribution for the instance constructed above.  We will also write $R(p; G)$ for the robust revenue guarantee of pricing $p$ for the corresponding problem instance, and similarly for $R(p,F; G)$ and $R^{*}(G)$.

We will now establish upper and lower bounds on $R^*(G)$ as a function of the size of the maximum independent set in $G$.  We begin with a lower bound on $R^*(G)$, which applies to any independent set (not just the maximum independent set).

\begin{lemma}
\label{lem.IS.pricing.lb}
If $G$ has an independent set $S$, then $R^*(G) \geq \tfrac{1}{2}\bigl(1-\frac{\sqrt{n}}{2^{n-1}}\bigr)|S|$.  
\end{lemma}
\begin{proof}
Given $S$, we will construct a pricing $p$ such that $R(p; G) \geq \tfrac{1}{2}\bigl(1-\frac{\sqrt{n}}{2^{n-1}}\bigr)|S|$.
For each item $i \in S$, choose price $p_i = 2^{i  n-1}$.   For each item $i \not\in S$, choose $p_i = +\infty$.  Let $F$ be any distribution of values compatible with the marginals $F_1, \dotsc, F_n$, and consider the revenue obtained under $F$.

Choose some item $i \in S$.  With probability $2^{-i  n}$ we will have $v_i = 2^{i  n}$, in which case the utility player $i$ obtains from buying item $i$ is $2^{i  n-1}$.  
Since $S$ is an independent set, there is no $j \in S$ such that $i$ and $j$ are adjacent in $G$, and therefore we cannot have $v_j = 2^{i  n}$ for any $j \neq i$.  Therefore, if item $i$ is not purchased, then it must be that there is some item $j \in S \backslash \{i\}$ with value $v_j = 2^{k  n}$ where $k > i$.  Taking a union bound over items, the probability of this event is at most $n \sum_{k > i} 2^{-kn}/\sqrt{n} < \sqrt{n} 2^{-(i+1) n+1}$.  Thus the probability that $v_i = 2^{i n}$ and item $i$ is purchased is at least $2^{-i n} - \sqrt{n} 2^{-(i+1) n+1} = 2^{-i n} \cdot (1 - \frac{\sqrt{n}}{2^{n-1}})$.

Taking a sum over all items in $S$, the total revenue obtained is therefore at least 
\[ \sum_{i \in S} p_i  2^{-i n} \cdot \left(1 - \tfrac{\sqrt{n}}{2^{n-1}}\right) = \left(1 - \tfrac{\sqrt{n}}{2^{n-1}}\right) \sum_{i \in S} \tfrac{1}{2} = \tfrac{1}{2}\left(1 - \tfrac{\sqrt{n}}{2^{n-1}}\right)|S|. \qedhere\]
\end{proof}

We next prove an upper bound on the max-min revenue $R^*(G)$ as a function of the size of the maximum independent set in $G$.  This direction is more subtle, as we must argue that \emph{no} pricing achieves more than the claimed revenue bound.

\begin{lemma}
\label{lem:MIS.pricing.ub}
If $M$ is a maximum independent set of $G$, then $R^*(G) \leq \bigl(|M|+2\bigr)\sqrt{n} + 3n\cdot 2^{-n}$.
\end{lemma}
\begin{proof}
Choose any pricing $p = (p_1, \dotsc, p_n)$.  We will partition the $n$ items into three sets, based on their price: set $L$ contains all items $i$ for which $p_i \leq 2^{(i-1)n}$.  Set $H$ contains all items $i$ for which $p_i > 2^{i n}$.  Set $S$ contains all items $i$ for which $p_i \in (2^{(i-1)n}, 2^{in}]$.  We think of $L$ as the items whose prices are much lower than their (single-item) revenue-maximizing price, and of $H$ as the items whose prices are higher than their revenue-maximizing price.  We will bound the maximum revenue obtainable from each of these three sets.

For each $i \in L$, the probability that $v_i > 0$ is at most $\sum_{j \geq i} 2^{-jn} < 2 \cdot 2^{-in}$.  So since $p_i \leq 2^{(i-1)n}$, the total revenue generated through sales of $i$ is at most $2^{(i-1)n}\cdot 2 \cdot 2^{-in} = 2 \cdot 2^{-n}$.  Since $|L| \leq n$, the total revenue generated from items in $L$ (in any distribution compatible with the marginals) is at most $2n \cdot 2^{-n}$.

Next consider $H$.  Choose some $i \in H$, and suppose that $p_i \in (2^{(j-1)n}, 2^{jn}]$ where $j > i$.  Then the probability that $v_i > p_i$ is at most $\sum_{k \geq j} 2^{-kn} / \sqrt{n} \leq 2 \cdot 2^{-jn} / \sqrt{n}$.  So since $p_i \leq 2^{jn}$, the total revenue generated through sales of $i$ is at most $2^{jn} \cdot 2 \cdot 2^{-jn} / \sqrt{n} = 2 / \sqrt{n}$.  Since $|H| \leq n$, the total revenue generated from items in $H$ (in any distribution compatible with the marginals) is at most $2n / \sqrt{n} = 2 \sqrt{n}$.

Finally we consider the set $S$ of all items $i$ for which $p_i \in (2^{(i-1)n}, 2^{i n}]$.  This is the most interesting case.  Note that if $i, j \in S$ with $i < j$, we must have $p_i < p_j$ (from the definition of $S$).  For each $i$, write $N(i) = \bigl\{ j < i \colon \{i,j\} \in E \bigr\}$.  That is, $N(i)$ is the set of neighbors of $i$ with lower index.  Write $T \subseteq S$ for the subset of nodes $T = \bigl\{ i \in S \colon |N(i) \cap S| < \sqrt{n} \bigr\}$.  That is, $T$ contains all nodes of $S$ that have fewer than $\sqrt{n}$ neighbors with lower index within $S$.

We claim that there is a distribution $F$ compatible with the marginals for which the revenue generated from the items in $S$ is at most $\bigl(1+o(1)\bigr) \cdot |T|$.  We define this (correlated) distribution by describing a process for sampling from the distribution.  First, choose at most one item $i \in [n]$ to have value $2^{in}$, consistent with the marginals (i.e., with probability $2^{-in}$).  If $i \in S \backslash T$, choose some $j \in N(i) \cap S$ uniformly at random and set $v_j = 2^{in}$, and set \emph{all} other values (including the values of items not in $S$) to $0$.  Note that since $i \in S \backslash T$ implies $\bigl|N(i) \cap S\bigr| \geq \sqrt{n}$, this process sets $v_j = 2^{in}$ with probability at most $2^{-in} / \sqrt{n}$, which is consistent with the marginal $F_j$.  In the event that $i \not\in S \backslash T$ or if no item $i$ has value $2^{in}$, values can be correlated arbitrarily.

Under this distribution $F$, when item $i \in S \backslash T$ has value $v_i = 2^{in}$, there is exactly one other item $j \in S$ with $j < i$ such that $v_j = 2^{in}$, and all other items have value $0$.  Since $j \in S$ we have $p_j < p_i$, so item $j$ will be sold.  Since we know that $p_j \leq 2^{(i-1)n}$, we conclude that the total revenue that can be generated from the event that $v_i = 2^{in}$ is at most $2^{(i-1)n} \cdot 2^{-in} = 2^{-n}$.  Further, the total probability that $v_i > 2^{in}$ is at most $\sum_{k > i} 2^{-kn}/\sqrt{n} \leq 2^{-(i+1)n+1}/\sqrt{n} \leq 2^{-(i+1)n}$, so the total revenue generated from sales of item $i \in S \backslash T$ due to events where $v_i > 2^{in}$ is at most $2^{in}\cdot 2^{-(i+1)n} = 2^{-n}$.

On the other hand, for any item $i \in T$, we know that $p_i \in (2^{(i-1)n}, 2^{in}]$ and the total probability that $v_i > 2^{(i-1)n}$ is at most $2^{-in} + \sum_{k > i}2^{-kn}/\sqrt{n} < 2^{-in} + 2^{-(i+1)n}$.  Thus the total revenue generated by selling item $i$ is at most $2^{in} \cdot (2^{-in} + 2^{-(i+1)n}) = 1 + 2^{-n}$.  Since $|S| \leq n$, we conclude that the total revenue obtained from sales of items in $S$ is at most $|T| + |T|2^{-n} + |S \backslash T|2^{-n} \leq |T| + n\cdot 2^{-n}$ as claimed.

Finally, we claim that $|T| \leq |M|\sqrt{n}$, where recall that $M$ is a maximum independent set in~$G$.  We will prove the claim by using $T$ to construct an independent set $M'$.  Start with $M' = \emptyset$.  Starting with the highest-indexed node from $T$, say $i$, add $i$ to $M'$ and remove $i$ and all of $i$'s neighbors from $T$.  From the definition of $T$, this removes at most $\sqrt{n}$ nodes from $T$.  Then take the highest-indexed node still in $T$, add it to $M'$, and remove it and its neighbors from $T$.  Repeat this process until $T$ is empty.  As we removed at most $\sqrt{n}$ nodes from $T$ on each step, we have $|M'| \geq |T|/\sqrt{n}$.  And by construction $M'$ is an independent set.  By maximality of $M$, we must therefore have $|M'| \leq |M|$, and hence $|M| \geq |T|/\sqrt{n}$.  Rearranging yields $|T| \leq |M|\sqrt{n}$ as claimed.

The total revenue obtained from sales of items in $S$ is therefore at most $|T| + n\cdot 2^{-n} \leq |M|\sqrt{n} + n\cdot 2^{-n}$.  Adding in the revenue contribution from $L$ and $H$ (and recalling that our bounds for $L$ and $H$ hold for any distribution and therefore for $F$), the total revenue generated under distribution $F$ is at most $|M|\sqrt{n} + n\cdot 2^{-n} + 2n\cdot2^{-n} + 2\sqrt{n} = (|M|+2)\sqrt{n} + 3n\cdot 2^{-n}$.
\end{proof}

We are now ready to complete the proof of \cref{thm:hardness}.  

\begin{proof}[Proof of \cref{thm:hardness}]
The hardness result for MIS directly implies the following: there is a class of graphs $\mathcal{G}$ in which the maximum independent set size is either less than $n^\delta$ or greater than $n^{1-\delta}$ and it is NP-hard to decide whether a given graph $G \in \mathcal{G}$ falls into the former category or the latter.

Given an algorithm $A$ for the max-min pricing problem, consider the following procedure for the MIS decision problem described above.  Given a graph $G \in \mathcal{G}$, let $p$ be the pricing returned by algorithm $A$ on input instance $(F_1(G), \dotsc, F_n(G))$.  Given $p$, compute the Adversary's best response distribution $F$ (using the algorithm from \cref{sec:adv-BR}) and then use this to compute $R(p, F; G) = R(p; G)$.  If $R(p; G) \geq 2n^{1/2 + \delta}$ our procedure declares that $G$ has an independent set of size at least $n^{1-\delta}$, otherwise it declares that its maximum independent set size is less than $n^\delta$.

We now claim that if the pricing algorithm $A$ has approximation factor at most $n^{1/2 - 3\delta}$, then the procedure above classifies graph instances correctly.  Suppose $G$ has an independent set $M$ of size at least $n^{1-\delta}$.  Then by \cref{lem.IS.pricing.lb}, $R^*(G) \geq \tfrac{1}{2}\bigl(1-\frac{\sqrt{n}}{2^{n-1}}\bigr)n^{1-\delta}$.  By the supposed approximation factor of the pricing algorithm, this means that $R(p;G) \geq \tfrac{1}{2}\bigr(1-\frac{\sqrt{n}}{2^{n-1}}\bigr)n^{1/2 + 2\delta} > 2n^{1/2 + \delta}$ for sufficiently large $n$.  The procedure therefore classifies such graphs $G$ correctly.

On the other hand, suppose that the maximum independent set $M$ of $G$ has size at most $n^{\delta}$.  Then by \cref{lem:MIS.pricing.ub}, $R^*(G) \leq (n^{\delta}+2)\sqrt{n} + 3n\cdot 2^{-n} < 2n^{1/2 + \delta}$ for sufficiently large $n$, and hence $R(p; G) < 2n^{1/2 + \delta}$ as well.  The procedure therefore classifies this class of graphs correctly as well, and is therefore correct in all cases.

We conclude that it is NP-hard to achieve an approximation factor better than $n^{1/2 - 3\delta}$ for any $\delta > 0$.  Setting $\epsilon = 3\delta$ completes the proof.
\end{proof}

\section{Uniform Pricing}
\label{sec:uniform}
In this section we study max-min pricing instances where setting a uniform pricing (all items have the same price) or even a single price (a single item is offered at less than $\infty$) is $\alpha$-max-min optimal for constant $\alpha$. 
This kind of pricing is desirable due to its great simplicity, and in the case of a single price, also its agnosticism of the correlation (see \cref{sub:prelim-problem-def}). Since uniform pricings dominate single prices, we will mostly prove our positive guarantees for single prices and our lower bounds for uniform pricings. This only serves to strengthen our results, although in practice the seller may prefer uniform pricing as this can only increase revenue.\footnote{For the same reason, and since tie-breaking only arises for uniform pricings and not for single prices, in this section we can assume essentially without loss of generality that the buyer breaks ties in favor of higher-priced items.}

\paragraph{Warm-up: Identical marginal distributions.} It is not hard to see that simple uniform pricing and in fact single prices arise naturally as the optimal max-min pricing in the case of symmetric marginals: When marginals are identical, the adversary can choose $\com$ as the correlation, such that all items have the same value at every value profile in the support. The seller can do no better than to use a single price in this case, since the buyer will always purchase a lowest-priced item. The max-min pricing is thus the Myerson monopoly price of the shared marginal distribution, set as a single price for an arbitrary item.
To summarize:
\begin{observation}
	\label{obs:id-case}
	For every setting with identical marginals there exists a single price that is max-min optimal. 
\end{observation}

We emphasize that the takeaway message from the identical marginals case is \emph{not} that the seller should price items at their individual monopoly prices---in fact this pricing strategy could have dire robust revenue guarantees even in extremely simple cases due to cannibalization.\footnote{For example, consider $m=2$ items with marginals $F_1=U[1,2]$ and $F_2=U[H,H+\frac{1}{2}]$ where $H$ can be arbitrarily large. Let $p=(1,H)$ be the pricing based on individual monopoly prices, then $R(p,\com)=1$ since the buyer always prefers item~1. However, for $p'=(\infty,H)$ the robust revenue guarantee is $R(p')=H$.}
This makes the unit-demand setting very different from the additive one. We expand upon this point in \cref{sub:uniform-regular}. 

\paragraph{Section overview.} 
In this section we explore the extent to which we can extend the approach of using just one price to get max-min (near-)optimality.
In \cref{sub:uniform-MHR} we show that if the marginals all have a \emph{monotone hazard rate (MHR)}, then a single price achieves a $\sim3.5$-approximation (roughly $30\%$ of the max-min optimum). 
The class of MHR marginals includes many well-studied distributions such as uniform, normal, exponential, and log-concave distributions. In \cref{sub:uniform-regular} we show the limitations of uniform pricing, by constructing an instance with \emph{regular} marginals for which such a pricing is no better than a $\Theta(n)$-max-min approximation. Our example highlights aspects in which (approximate) max-min optimal pricing can be complex.

\paragraph{Section preliminaries.}
In this section it will be convenient to assume that marginals are continuous and differentiable distributions with density functions denoted by $f_i$ for distribution $F_i$.\footnote{The results hold for discretized versions as well, with appropriate definitions of discrete MHR and discrete regularity as have already been developed in the literature (see, e.g., \cite{Elkind07}).}
We further assume that $F_i$ is strictly increasing for every $i$. The inverse function $F_i^{-1}(\cdot)$ is thus well-defined, and for every $q\in[0,1]$ the value $\mu^q_i=v_i(q_i)$ at quantile $q$ is also well-defined (namely, $\mu^q_i=F^{-1}_i(q)$).
In \cref{sub:uniform-regular} we also consider \emph{truncated} versions of such distributions, in which we allow a nonzero probability mass $\rho$ on $v_i^{\max}$.
The inverse function $F_i^{-1}(\cdot)$ of a truncated distribution $F_i$ can still be easily defined: if $q\in[1-\rho,1]$ then let $F_i^{-1}(q):=v_i^{\max}$. 

\subsection{Constant-Factor Approximation for MHR Marginals}
\label{sub:uniform-MHR}

A distribution $F$ with density $f$ is \emph{MHR} if its hazard rate  $\frac{f(v)}{1-F(v)}$ is (weakly) increasing as a function of the value $v$.
Our starting point is the following observation, which shows we cannot hope for better than a constant factor approximation:

\begin{observation}
	\label{obs:uniform-not-opt}
	There is an instance with $n=2$ items and MHR marginals (in fact, uniform distributions) such that for some constant $\alpha> 1$, no uniform pricing is $\alpha$-max-min optimal.
\end{observation}
\noindent
The proof is in \cref{appx:MHR}, by considering uniform marginals $U[\frac{1}{4},\frac{1}{4}+\epsilon]$ and $U[0,1]$.

Our main result in this section is that for MHR marginals, setting the maximum of the medians of the marginals as a single price (for the item with the corresponding marginal) is approximately-max-min optimal up to a small constant factor (recall that in practice one may prefer to use this price for all items uniformly):

\begin{theorem}[Max median as a single price]
	\label{mhr}
	Consider a setting with MHR marginals and let $\mu_{\max}$ be the maximum of the medians. Then a single price $p$ of $\mu_{\max}$ for an item with median $\mu_{\max}$ is $3.443$-max-min optimal, and achieves a robust revenue guarantee $R(p) \ge \frac{1}{3.443}\OPT$ against any joint distribution, where $\OPT$ is the expected \emph{welfare} with joint distribution $\com$.
\end{theorem}

In \cref{sub:MHR-intuition} we provide intuition for pricing using the maximum median.
We remark that an advantage of such pricing is that even if the marginals for individual items are not fully known, the price can be easily estimated from a small number of samples \cite{AzarDMW13} (this can be done with separate samples for individual items, as recall that there is no access to the joint distribution).  
Another immediate implication of \cref{mhr} is that, since the approximation guarantee for MHR marginals is achieved with respect to expected \emph{welfare}, 
the suggested pricing yields a constant approximation not only to the robust revenue of the best \emph{pricing}, but also to the robust revenue of the best \emph{mechanism} (which could be complex and/or randomized).

\subsubsection{Proof of Theorem~\ref{mhr}}
\label{sub:MHR-approx}

We begin by stating a property of MHR distributions that will be useful in the proof. 
This property reflects the fact that the exponential distribution is the ``extreme'' MHR distribution in the sense of lowest hazard rates and thus heaviest tail. Given an upper bound of~$x$ on the value at quantile $q$ of an MHR distribution $F$, the next lemma provides an upper-bound on its value at quantile $q'>q$, using the value at $q'$ of an exponential distribution with parameter $\lambda=-\frac{\ln(1-q)}{x}$ (i.e., the exponential distribution that at quantile $q$ has value $x$).

\begin{lemma}[MHR property]
	\label{lem:MHR-vals-above-q}
	Let $F$ be an MHR distribution and let $q\in[0,1]$ be a quantile. 
	If $F^{-1}(q)\le x$ then
	$
	\forall q'>q : F^{-1}(q') \le \frac{x\ln (1-q')}{\ln(1-q)}.
	$
\end{lemma}

Similar properties to the one in \cref{lem:MHR-vals-above-q} appear in the literature (see, e.g., \cite{cr17});
we include a proof for completeness in \cref{sub:MHR-vals-above-q}.
\cref{mhr} follows directly from the next lemma:

\begin{lemma}
	\label{mhr-comon}
	For every $q\in[0,1]$, consider MHR marginals with values $\mu^q_1,\dots,\mu^q_n$ at quantile~$q$, and let $\mu=\mu^q_{\max}=\max_{i}\{\mu^q_i\}$ be the highest such value. 
	When the joint distribution is $\com$, setting $\mu$ as a single price for an item $i$ with $\mu_i^q=\mu$ achieves expected revenue that is a $(\frac{1}{1-q} -\frac{1}{\ln(1-q)})$-approximation to the expected welfare.
\end{lemma}

\begin{proof}[Proof of \cref{mhr}]
	For every single price $p$, we have that $R(p,F)=R(p,\com)$ for every compatible~$F$, and so $R(p)=R(p,\com)$.
	Let $\OPT$ be the expected welfare from allocating to the unit-demand buyer when the joint distribution is $\com$; then	for every pricing $p'$ we have that $R(p',\com)\le \OPT$, and we conclude that $R^*\le \OPT$. 
	By \cref{mhr-comon} with $q$ set to $\frac{1}{2}$, pricing an item with highest median by this median is guaranteed to yield revenue of $(2 +\frac{1}{\ln2})\cdot\OPT$ against $\com$. We denote the resulting single price by $p^*$. So
	$
	R(p^*)=R(p^*,\com)\ge (2 +\frac{1}{\ln2})\cdot\OPT\ge (2 +\frac{1}{\ln2})\cdot R^*.
	$
	The \lcnamecref{mhr} follows.
\end{proof}

\begin{proof}[Proof of \cref{mhr-comon}]
	Assume throughout that the joint distribution is $\com$. 
	Denote by $\ALG$ the expected revenue extracted by setting $\mu$ as a single price. Observe that
	\begin{equation}
	\ALG = (1-q)\mu.\label{eq:ALG}
	\end{equation}
	Denote the expected welfare by $\OPT$, and observe
	\begin{equation}
	\OPT=\mathbb{E}_{\xi\sim U[0,1]} \bigl[\max_i\{\mu^{\xi}_i\}\bigr] = \int_0^q \max_i\{\mu^\xi_i\} d\xi + \int_q^1 \max_i\{\mu^\xi_i\} d\xi \le q\mu + \int_q^1 \max_i\{\mu^\xi_i\} d\xi.\label{eq:OPT}
	\end{equation}

	Our goal is thus to upper-bound $\int_q^1 \max_i\{\mu^\xi_i\} d\xi$. Using that $F_i^{-1}(q)\le \mu$ for every $i$ and applying \cref{lem:MHR-vals-above-q}, we get that
	for every quantile $\xi>q$ the following holds: 
	\begin{equation} 
	\forall i : \mu^{\xi}_i= F_i^{-1}(\xi) \le \tfrac{\mu\ln (1-\xi)}{\ln(1-q)}=\mu^\xi_{\Exp[\lambda]}~,\label{eq:UB-mu-xi}
	\end{equation} 
	where $\mu^\xi_{\Exp[\lambda]}$ is the value at quantile $\xi$ of the exponential distribution $\Exp[\lambda]$ with parameter $\lambda=-\frac{\ln(1-q)}{\mu}$, i.e., 
	$
	\mu^\xi_{\Exp[\lambda]} =	-\frac{\ln(1-\xi)}{\lambda}$.

	From \cref{eq:UB-mu-xi} it follows that $\int_q^1 \max_i\{\mu^\xi_i\} d\xi$ is upper bounded by $\int_q^1 \mu^\xi_{\Exp[\lambda]} d\xi$. Notice the latter is the contribution to the expectation of $\Exp[\lambda]$ from values above $\mu^q_{\Exp[\lambda]}$. By a standard calculation (due to the memorylessness of the exponential distribution) this is $\frac{(1-q)(1-\ln(1-q))}{\lambda}=(1-q)\mu-\frac{1-q}{\ln(1-q)}\mu$. Plugging this into \cref{eq:OPT} and using \cref{eq:ALG} we get
	$$
	\OPT \le q\mu + (1-q)\mu-\tfrac{1-q}{\ln(1-q)}\mu =\mu -\tfrac{\ALG}{\ln(1-q)}=\left(\tfrac{1}{1-q} -\tfrac{1}{\ln(1-q)}\right)\cdot\ALG.
	$$ 
	This completes the proof of the lemma.
\end{proof}

\paragraph{Discussion of the tightness of \cref{mhr}.}
The robust revenue guarantee of $3.443$ is tight for pricing at the maximum of the medians~$\mu$; to see this, take as marginals $U[\mu-\epsilon,\mu+\epsilon]$ and $\Exp[\frac{\ln 2}{\mu}]$. Similarly, the guarantee in \cref{mhr-comon} is tight for any $q\in[0,1]$. Thus, \cref{mhr-comon} shows that pricing at the maximum of quantiles $\mu_1^q,\dots,\mu_n^q$ for any $q$ cannot do much better than using $q=\frac{1}{2}$.\footnote{
The guarantee of $\frac{1}{1-q} -\frac{1}{\ln(1-q)}$ is minimized at $q^*=1-e^{-2W(1/2)}$ where $W$ is the Lambert (product log) function. Setting $q=\frac{1}{2}$ achieves a ratio of $3.443$, which is only slightly higher than the optimal one achieved by $q^*$.}
We leave as an open question whether there is some other single price or uniform pricing rule that improves upon $3.443$ more significantly.

\subsection{Lower Bound Beyond MHR}
\label{sub:uniform-regular}

Our positive result for MHR marginals in the previous section begs the question of whether it can be extended beyond MHR, in particular to regular marginals. 
A distribution $F$ with density $f$ is \emph{regular} if its \emph{virtual value} $\varphi(v)= v-\frac{1-F(v)}{f(v)}$ (the value minus inverse the hazard rate) is (weakly) increasing as a function of $v$.
In this section we give a strong negative answer.

\begin{proposition}[Linearly-many prices required]
	\label{pro:single-nonopt-cont}
	There exists an instance of the max-min pricing problem with $n$ regular marginals such that for every $k\in[n]$, no pricing with $\le k$ different prices is $o(\frac{n}{k})$-max-min-optimal, and there is a pricing with $k$ different prices that is $\Theta(\frac{n}{k})$-max-min-optimal.
\end{proposition}

An immediate corollary is the following:

\begin{corollary}[Uniform pricing far from max-min optimal]
	There exists an instance of the max-min pricing problem with $n$ regular marginals for which no uniform pricing is $o(n)$-max-min optimal. 
\end{corollary}

\cref{pro:single-nonopt-cont} and its corollary show that unlike in the MHR case, in the regular case there can be a linear gap between that revenue that is achievable with uniform versus general pricing, and in fact the robust revenue guarantee may increase linearly as more distinct prices are used. 

Our construction of an instance that proves \cref{pro:single-nonopt-cont} appears in \cref{ex:single-nonopt-cont}. 
The construction uses marginals with an exponentially wide range of values and with a Myerson revenue of~$1$ across all items---we explain in \cref{appx:beyond-MHR} why this can be expected. The construction makes use of the \emph{equal revenue} distribution $\eqrev$, defined as $\eqrev(v)=1-\frac{1}{v}$ for $v\in[1,\infty)$, after truncating it at different points. A salient property of this distribution (including its truncated versions) is that in a single-item setting, taking any value $v$ in the support as the item's price yields the same expected revenue of~$1$ for the seller (since the revenue is $v\cdot\bigl(1-F(v)\bigr)=1$ for every $v$). 

\begin{example}[Truncated equal-revenue marginals]
	\label{ex:single-nonopt-cont}
	Let $t$ be a vector of $n$ truncation points $(t_1,\dots,t_n)$ where $t_j=2^{j+1}$ for every $j\in[n]$.   
	Let the $n$ marginals be as follows: For every item $j\in[n]$, let $F_j$ be $\eqrev$ truncated at $t_j$, that is, $F_j(v)=1-\frac{1}{v}$ for $v\in[1,t_j)$ and $F_j(t_j)=1$.
\end{example}

While as defined the marginals in \cref{ex:single-nonopt-cont} are not continuous, they can be smoothed to have no point masses (by ``smearing'' the mass at the truncation point uniformly over a vanishingly small interval), resulting in regular distributions. An alternative to smoothing is extending the definition of regular distributions to those with truncation points, by simply defining the virtual value at a truncation point $t_j$ to be $t_j$ itself, and requiring monotonicity of the virtual value function as in the standard definition of regularity.  

The main component in the analysis of \cref{ex:single-nonopt-cont} is the following lemma:

\begin{lemma}[Pricing at half the threshold]
	\label{lem:half-thresh}
	Consider \cref{ex:single-nonopt-cont} and fix a subset $S\subseteq [n]$ of items. Let $p$ be the pricing of every item $j\in S$ at price $\frac{t_j}{2}=2^j$, and
	every item in $\overline{S}$ at $\infty$. Then $R(p)=\Theta(|S|)$, and $p$ is $\Theta(\frac{n}{|S|})$-max-min optimal.
\end{lemma}

See \cref{appx:beyond-MHR} for the proof of \cref{lem:half-thresh} as well as of \cref{pro:single-nonopt-cont}. We end this section with several takeaways from \cref{ex:single-nonopt-cont}. It is informative to compare the (approximate) max-min pricing for this example to the max-min pricing for an additive buyer, namely, pricing every item at its monopoly price ~\cite{Carroll17}. In \cref{ex:single-nonopt-cont}, any value in the item's value range is a monopoly price. But it turns out that to get approximate max-min optimality, the pricing needs to be ``spread-out'', in the sense that the probability of each consecutive item's value to exceed its price decreases by a factor of 2 (cf., \cite{BeiGLT19}). 
This kind of pricing cannot be computed separately for the different marginals, nor is it simply a function of the monopoly prices pf the marginals.
Another consequence of \cref{ex:single-nonopt-cont} that stands in contrast to the additive case is the following nonmonotonicity property of the max-min optimal pricing; in the mechanism design literature, nonmonotonicity is usually taken as a sign of complexity (e.g.,~\cite{HartReny}):

\begin{corollary}[Nonmonotonicity of max-min revenue in the marginals] 
	\label{cor:nonmonotone}
	There exist two max-min pricing instances with $n$ items each and optimal robust revenue guarantees $R^*_1$ and $R^*_2$, respectively, such that for every item $j$, the marginal $F_j^1$ of $j$ in the first setting first-order stochastically dominates the marginal $F_j^2$ of $j$ in the second setting, yet $R^*_1\ll R^*_2$. In fact, the ratio $R^*_2/R^*_1$ can be as large as $\Omega(n)$.
\end{corollary}

\begin{proof}
	Let the first instance have identical marginals $\eqrev$ for all items, and let the second instance have truncated equal-revenue marginals as in \cref{ex:single-nonopt-cont}. Clearly, $\eqrev$ first-order stochastically dominates its truncated version for every $j\in[n]$. From \cref{obs:id-case} we know that $R^*_1=1$ (the Myerson revenue of the equal-revenue distribution), and $R^*_2=\Theta(n)$ by \cref{lem:half-thresh}. This completes the proof. 
\end{proof}

\bibliographystyle{plainnat}
\bibliography{abb,bib}

\appendix

\section{Proofs Omitted from Section \ref{sec:adv-BR}}
\label{appx:proofs-alg}
\begin{proof}[Proof of \cref{cor:alg}]
	Suppose towards contradiction that this is false, and let $i$ be the smallest index contradicting this. That is, the coupling $c$ obtained by \cref{alg:adv-BR} simultaneously realizes $K_{[j]}$ for every $j<i$ (i.e., $k_j(c)=K_{[j]}-K_{[j-1]}$), but $k_i(c)<K_{[i]}-K_{[i-1]}$.
	However, by \cref{lem:residual-optimality}, 
	$k_i(c) = \max\bigl\{k_i(c') ~\big|~ \text{$c'$ is a coupling} \And k_j(c')=k_j(c)=K_{[j]}-K_{[j-1]} \ \forall j<i\bigr\}$, and by \cref{prop:sim-max}, there exists a coupling $c'$ such that $k_j(c')=K_{[j]}-K_{[j-1]}$ for every $j<i$ and $k_i(c') = K_{[i]}-K_{[i-1]}$; a contradiction.
\end{proof}

\begin{proof}[Proof of \cref{lem:residual-optimality}]
	One can easily verify Property (1). Indeed, by design, the \cref{alg:adv-BR} finds the maximum possible number of chains dominated by item 1. 
	To prove Property (2), note that in order to satisfy $k_j(c')=k_j(c)$ for every $j<i$, it must be that $\sum_{j<i}k_j(c)$ utilities of each item $i, \ldots, n$ are already coupled (in iterations $j<i$). To maximize the number of chains dominated by item $i$, it is best if the coupled utilities of items $>i$ are largest and those of item $i$ are smallest. This is exactly what \cref{alg:adv-BR} does.
\end{proof}

\begin{proof}[Proof of \cref{lem:add-1}]
	First, observe that by the fact that $k'_1<K_1$, it holds that $u_i^d \prec u_1^{k'_1+1}$ for every $i \geq 2$ (else the number of chains dominated by item 1 in any coupling is at most $k'_1$).
	We distinguish between two cases, based on the coupling before the sift\&lift process. 
	Case 1: the lowest chain is rooted at item 2. Then, its removal vacates utility $u_i^d$ for every $i\geq 3$ and the lemma follows. 
	Case 2: the lowest chain is rooted at item 1. Suppose by way of contradiction that after the sift\&lift process there exists an item $i\geq 3$ such that $u_i^d$ is coupled. The chain rooted at $u_2^{K_{[2]}-k'_1}$ vacated some utility $u_i^{\ell}$ of item $i$. Consider the (consecutive) utilities of item 1 following $u_2^{K_{[2]}-k'_1}$ in the set $S$, $u_1^q \geq \cdots \geq u_1^{k'_1}$. Some (possibly 0) of the chains rooted at these utilities were lifted to utilize a higher utility of item $i$ as a result of the removal of the bottom chain rooted at 2, but the bottom chain rooted at 1 was not (else, $u_i^d$ would be available). Let $u_1^r$ be the highest utility of item 1 among these utilities ($k'_1 \geq r \geq q$) that was not lifted, and let $u_i^{m}$ be the utility of item $i$ in the chain rooted at $u_1^r$. By the choice of $u_1^r$, we have that $u_i^{m-1}$ must be available. Thus, the fact that the chain rooted at $u_1^r$ was not lifted means that $u_i^{m-1}$ is not dominated by $u_1^r$. This means 
	that there are at most $k_1'-r+1$ utilities of item $i$ that are respectively (in some order, without repetition) dominated by utilities out of $u_1^r, \ldots, u_1^{k_1'}$. This, in turn, means that there are at most $k'_1$ utilities of item $i$ that are respectively dominated by utilities of item 1, contradicting $k'_1 < K_1$.
\end{proof}

\section{The Sift\&Lift Process (Algorithm \ref{alg:sift-lift})}
\label{appx:sift-lift}
In this section we present the sift\&lift process.

\begin{algorithm}
	\caption{The sift\&lift process; Input: utilities $u_1,\ldots,u_n$; (partial) coupling $c$.}
	\label{alg:sift-lift}
	\begin{algorithmic}[1]
		\STATE Let $t$ be the chain in $c$ rooted at $u_2^{K_{[2]}-k'_1}$
		\STATE $c \gets c \setminus \{t\}$ \qquad \COMMENT{i.e., decouple all utilities in chain t}
		\STATE Let $q$ be the index of the utility of item 1 following $u_2^{K_{[2]}-k'_1}$ in $u_{1,2}$
		\IF{$u_i^q \in c$}
		\FOR{$j=q, \ldots, k_1'$} 
		\STATE decouple all utilities in the chain rooted at $u_1^j$
		\FOR{$i=3, \ldots, n$} 
		\STATE let $\ell_i =  \arg\max_{\ell}\{u_i^{\ell} \mid u_i^{\ell}\not\in c \mbox{ and } u_i^{\ell}\prec u_1^j\}$ 
		\COMMENT{break ties towards a lower index}
		\ENDFOR 
		\STATE $c \gets \bigl(u_1^j,\{u_i^{\ell_i}\}_{i \geq 3}\bigr)$ \COMMENT{add lowest uncoupled utility of item 2}
		\ENDFOR
		\ENDIF
		\RETURN $c$
	\end{algorithmic}
\end{algorithm}

\section{The Adversary's Best Response: Nondiscrete Distributions}
\label{appx:alg-extensions}
In this \lcnamecref{appx:alg-extensions} we consider the adversary's best response when the distributions are not discrete. In this case the mechanics of the adversary's best response are essentially the same, however as is often the case with nondiscrete distributions, we run into problems of certain suprema and infima not necessarily being attainable. To see this, consider two items, one whose marginal distribution is uniform in $[1,2]$ and one whose marginal distribution is uniform in $[4,5]$. Consider pricing the first item at $1$ (so its utility is uniform in $[0,1]$) and the second item at $4$ (so its utility is also uniform in $[0,1]$).
To analyze the adversary's behavior, let us also assume that tie-breaking among the same utilities is in favor of the second (higher-priced) item.\footnote{This is the most commonly assumed tie-breaking rule, since as the seller is the mechanism designer and any tie-breaking chosen by the seller is still incentive-compatible for the buyer, it is generally assumed that the seller opts for the tie-breaking rule that maximizes her revenue.}

Let us first analyze the infimum of the revenues attainable by the adversary. We note that for every $\varepsilon>0$, the adversary can make the sale probability of the first item as high as $1-\varepsilon$ and the same probability of the second item as low as $\varepsilon$, by coupling each value $v\in[1+\varepsilon,2]$ of item $1$ with value $v-\varepsilon$ of item $2$ (and then, say, coupling each remaining value $v[1,1+\varepsilon$] of item $1$ with value $3+v+1-\varepsilon$ of item $2$), hence coupling each utility $u\in[\varepsilon,1]$ from item $1$ with utility $u-\varepsilon$ of item~$2$. This results in a  revenue of $1+4\cdot\varepsilon$, and so the infimum of revenues achievable by the adversary is $1$. (It is impossible for the adversary to achieve even lower revenue since for these prices with probability $1$ some item would be sold, and hence for any correlation the expected revenue must be at least the lowest price: $1$.)

We will now observe that there is nonetheless no correlated distribution with the given marginals that gives a revenue of $1$, as giving such a revenue would mean selling item $1$ with probability $1$, which means that the two utility distributions, which are both uniform on $[0,1]$, would have to be coupled such that the utility from item $1$ is with probability $1$ \emph{strictly} (due to the tie-breaking rule) higher than the utility from item $2$, which is impossible to achieve since the marginal distributions of both utilities are the same.\footnote{To see this, assume for contradiction that such a correlated distribution $F$ exists. Then by linearity of expectation and definition of the marginals, $\mathbb{E}_F[u_1-u_2]=\mathbb{E}_F[u_1]-\mathbb{E}_F[u_2]=\mathbb{E}_{U[0,1]}[u_1]-\mathbb{E}_{U[0,1]}[u_2]=0$. But remember that $u_1-u_2$ is an almost-surely-nonnegative random variable by assumption. Since its expectation is zero, we have that for every $m\in\mathbb{N}$ the probability that $u_1-u_2\ge\nicefrac{1}{m}$ is zero. Therefore, by $\sigma$-additivity, we have that the probability that $u_1-u_2>0$ is zero---a contradiction.}
Note that indeed the limit of the above couplings, as $\varepsilon$ tends to $0$, is the identity coupling, which sells item $2$ with probability $1$, and hence gives revenue $4$ (a discontinuity).

The above example shows that without going into model details that we have abstracted away, such as the tie-breaking rule (and specifically, by our example, without disallowing the standard tie-breaking in favor of higher-priced items), it is hopeless to expect the existence of a precise best-response for the adversary rather than a only sequence of responses whose revenues converge to the infimum of the revenues that the adversary can attain. Once such a sequence is the most we can hope for, we will be content with noting that one can obtain a response of the adversary that gets arbitrarily close to the infimum by simply discretizing the distribution to a sufficiently fine grid, and then using our algorithm for the discrete case. We will not dive deeper into this direction as it bears no new conceptual messages, and anyway the major building block that we need for our hardness approximation from the next section is \cref{thm:alg-is-opt}: the correctness of the algorithm for discrete marginals, as well as its polynomial computational complexity in the size of its input.

\section{Proofs Omitted from Section \ref{sec:uniform}}
\label{appx:uniform}
\subsection{MHR Marginals}
\label{appx:MHR}

\begin{proof}[Proof of \cref{obs:uniform-not-opt}]
	Consider uniform marginals $U[\frac{1}{4},\frac{1}{4}+\epsilon]$ and $U[0,1]$ for items 1 and~2, respectively, where $\epsilon$ is sufficiently small. The robust revenue guarantee of any uniform pricing is $\le \frac{1}{4}+\epsilon$, since the seller either sets both prices to $\le \frac{1}{4}+\epsilon$, or effectively sets a single price for item 2, and the Myerson expected revenue of item 2 is $\frac{1}{4}$. 
	Consider now nonuniform prices $p_1=\frac{1}{4}$ and $p_2=\frac{5}{8}-\epsilon$.
	For any compatible distribution $F$, this pricing $p=(p_1,p_2)$ guarantees revenue of at least $\frac{1}{4}$ for every valuation profile in the support, and for every profile such that $v_2>p_2+\epsilon$ it guarantees $\frac{5}{8}-\epsilon$ (since for such profiles the utility from buying item 2 exceeds the utility from buying item~1 even if item~1's value is $v_1^{\max}$). 
	The robust revenue guarantee of $p$ is thus $\ge \Pr_{v\sim F}[v_2>\frac{5}{8}]\cdot(\frac{5}{8}-\epsilon) + (1-\Pr_{v\sim F}[v_2>\frac{5}{8}])\cdot\frac{1}{4}= \frac{3}{8}(\frac{5}{8}-\epsilon)+\frac{5}{8}\cdot \frac{1}{4}\to \frac{25}{64}$ as $\epsilon\to 0$. Comparing this to the upper bound of $\frac{1}{4}+\epsilon$ for uniform pricing completes the proof.
\end{proof}

\subsubsection{Intuition for pricing using the max median}
\label{sub:MHR-intuition}

\paragraph{``Standard'' uniform distributions.}

We build intuition by considering ``standard'' uniform marginals of the form $F_i=U[0,b_i]$, where without loss of generality $b_1\le \dots\le b_n$. Such marginals have the property that if the joint distribution is $\com$ then the highest value of any valuation profile in the support is~$v_n$. Thus the highest value has an MHR distribution (namely, $F_n=U[0,b_n]$), and moreover the expected welfare from allocating to the unit-demand buyer given $\com$ is its expectation $\mathbb{E}_{F_n}[v_n]$.
The expected welfare is clearly an upper-bound on the expected revenue $R(p,\com)$ for every pricing $p$, so $\forall p : R(p)\le \mathbb{E}_{F_n}[v_n]$ and we conclude that 
\begin{equation}
R^*\le \mathbb{E}_{F_n}[v_n].\label{eq:UB_on_Rstar}
\end{equation}

A salient property of MHR distributions is that in a single-item setting, the seller can extract a constant fraction of the expected welfare as revenue by setting either the Myerson monopoly price, or a price based on one of the quantiles, in particular the median \cite{DhangwatnotaiRY15}. 
So by using the median of~$F_n$ as a single price for item $n$, the seller is able to extract a constant fraction $\alpha$ of $\mathbb{E}_{F_n}[v_n]$, and thus by Eq.~\eqref{eq:UB_on_Rstar} also of $R^*$, when the joint distribution is $\com$. 
Due to the correlation agnosticism property of single prices, this $\alpha$-approximation guarantee holds for any joint distribution, establishing that using the highest median as a single price is $\alpha$-max-min optimal in this case.\footnote{In fact, for this particular case of ``standard'' uniform marginals, a single price is precisely max-min optimal---see \cref{pro:standard-uniforms} in \cref{sub:simple-case}.}

We remark that the exact same argument holds for \emph{exponential} marginals, since the marginal with the lowest rate-parameter $\lambda$ dominates the others. 

\paragraph{``Non-standard'' uniform distributions.}
The above argument cannot be extended to general uniform marginals of the form $F_i=U[a_i,b_i]$ where $a_i$ can be strictly positive, because the maximum value given the joint distribution $\com$ is not necessarily distributed according to an MHR distribution.\footnote{For example, consider $m=2$ items with marginals $F_1=U[\frac{1}{2}-\epsilon,\frac{1}{2}+\epsilon]$ and $F_2=U[0,1]$. The highest value given $\com$ is uniform over $[\frac{1}{2}-\epsilon,\frac{1}{2}]$ with probability $\frac{1}{2}$, and otherwise uniform over $[\frac{1}{2},1]$. So its CDF $F_{\max}(v)$ is $\frac{v-(1/2)+\epsilon}{2\epsilon}$ below $\frac{1}{2}$ and $v$ above it. This is not an MHR distribution since the hazard rate at $v<\frac{1}{2}$ is $\frac{1/2\epsilon}{1/2 + ((1/2)-v)/2\epsilon}\ge\frac{1}{2\epsilon}$, and above $\frac{1}{2}$ but below $\frac{3}{4}$ it is $\frac{1}{1-v}\le 4$.}
However, we argue that using the highest median $\mu_{\max}$ as a single price still guarantees a constant-factor approximation to $R^*$: The probability of the buyer purchasing the item with the highest median at this price is $\frac{1}{2}$, yielding expected revenue of $\frac{\mu_{\max}}{2}$ given any joint distribution. By definition $\mu_{\max}$ is at least the median $\mu_{i^*}$ of distribution $F_{i^*}$ where $i^*=\arg\max_i\{b_i\}$ (i.e., the distribution with highest $b_i$), and this median in turn is at least $b_{i^*}/2$. On the other hand, since $b_{i^*}=\max_i\{v_i^{\max}\}$, it is a clear upper bound on the expected welfare from allocating to the unit-demand buyer given any joint distribution. Putting everything together, if $p$ is the single price $\mu_{\max}$ for the item with the highest median, then
$$
R(p)\ge \frac{\mu_{\max}}{2}\ge \frac{\mu_{i^*}}{2} \ge \frac{b_{i^*}}{4} \ge \frac{R^*}{4}.
$$

\subsubsection{Proof of Lemma~\ref{lem:MHR-vals-above-q}} 
\label{sub:MHR-vals-above-q}

\paragraph{Hazard rate.}

Towards proving \cref{lem:MHR-vals-above-q}, we denote by $h(x)$ the hazard rate of a distribution~$F$ with density $f$ at value $x$, i.e., $h(x)=\frac{f(x)}{1-F(x)}$. 
Note that for the exponential distribution with rate-parameter $\lambda$, $h(x)=\lambda$.
We shall make use of the following general connection between hazard rate and CDF:
\begin{equation}
\int_{0}^{x} h(v)dv=-\ln\bigl(1-F(x)\bigr).\label{eq:HR-to-CDF}
\end{equation} 

The following claim lower-bounds the hazard rate of a distribution by the hazard rate $-\frac{\ln(1-q)}{\lambda}$ of an exponential distribution with the same (or higher) value at a given quantile $q$. 

\begin{claim}
	\label{cla:LB-on-HR}
	Let $F$ be an MHR distribution with density $f$, and let $q\in[0,1]$ be a quantile. 
	If $F^{-1}(q)\le x$ then the hazard rate of $F$ at $x$ is $h(x)\ge -\frac{\ln(1-q)}{x}$.
\end{claim}

\begin{proof}
	Assume for contradiction that $h(x)< -\frac{\ln(1-q)}{x}$. Then by monotonicity of the hazard rate~$h(\cdot)$, for every $v\le x$ we have $h(v)< -\frac{\ln(1-q)}{x}$. Thus $\int_{0}^{x} h(v)dv < -\ln(1-q)$, and by \cref{eq:HR-to-CDF},
	$$
	\ln(1-F(x))>\ln(1-q).
	$$
	Taking $\exp(\cdot)$ of both sides and rearranging we get $q>F(x)$, in contradiction to the assumption that $F^{-1}(q)\le x$.
\end{proof}

We now use \cref{cla:LB-on-HR} to prove \cref{lem:MHR-vals-above-q}.

\begin{proof}[Proof of \cref{lem:MHR-vals-above-q}]
	Fix $q'>q$. 
	By \cref{cla:LB-on-HR}, $h(x)\ge -\frac{\ln(1-q)}{x}$, and by monotonicity this inequality holds for every $v\ge x$. So 
	\begin{equation}
	\int_{F^{-1}(q)}^{F^{-1}(q')}h(v)dv\ge -\frac{(F^{-1}(q')-x)\ln(1-q)}{x} = -\frac{F^{-1}(q')\ln(1-q)}{x}+\ln(1-q),\label{eq:one-hand}
	\end{equation}
	where the inequality uses that $F^{-1}(q)\le x$.
	On the other hand, by \cref{eq:HR-to-CDF},
	\begin{equation}
	\int_{F^{-1}(q)}^{F^{-1}(q')}h(v)dv = -\ln(1-q') + \ln(1-q).\label{eq:other-hand}
	\end{equation}
	Putting \cref{eq:one-hand,eq:other-hand} together, 
	$$
	\frac{F^{-1}(q')\ln(1-q)}{x}\ge \ln(1-q').
	$$
	Rearranging (while taking into account that $\ln(1-q)<0$) completes the proof.
\end{proof}

\subsubsection{Max-min optimality of a single price in a simple case}
\label{sub:simple-case}

\begin{proposition}
\label{pro:standard-uniforms}
A single price is max-min optimal for uniform marginals of the form $U[0,b_j]$. 
\end{proposition}

\begin{proof}
	Consider $\com$ as the joint distribution. We show that a single price achieves the optimal expected revenue against $\com$. Recall (see \cref{ftn:sufficient}) that this is enough to establish max-min optimality of this single price.
	
Consider the optimal pricing given $\com$, excluding prices of $\infty$ and the corresponding items (if there are several optimal pricings take one with a minimum number of finite prices). 
We have at most $n$ prices $p_1,p_2, \dots$ (where $p_j$ is item $j$'s price); assume without loss of generality that the items are numbered such that $b_1\le b_2\le \cdots$. 
Denote the quantiles corresponding to the prices by $q_1,q_2,\dots$. 
Assume for simplicity that the buyer breaks ties in favor of higher-indexed items. 
So $q_1 < q_2 < \cdots$ (otherwise an item with lower $b_j$, or same $b_j$ but lower index, will never be bought), and therefore $p_1 < p_2 < \cdots$ (as a consequence of the monotonicity of $q$ and of $b$). 

Assume for contradiction that $p_1,p_2<\infty$. Since the joint distribution is $\com$, item 1 is bought from quantile $q$ such that $qb_1=p_1$ (i.e.,~$q=\frac{p_1}{b_1}$), item 2 is bought from quantile $q$ such that $qb_2-p_2=qb_1-p_1$ (i.e.,~$q=\frac{p_2-p_1}{b_2-b_1}$), and so on. The expected revenue from selling items 1 and 2 is
$$
\left(\frac{p_2-p_1}{b_2-b_1}-\frac{p_1}{b_1}\right)p_1+\left(\frac{p_3-p_2}{b_3-b_2}-\frac{p_2-p_1}{b_2-b_1}\right)p_2,
$$
where $\frac{p_3-p_2}{b_3-b_2}$ is understood to be 1 if there are only $m=2$ items.

Now fix $p_2,p_3$ and find $p_1$ that maximizes the revenue by taking the derivative:
$$
\frac{-p_1}{b_2-b_1}-\frac{p_1}{b_1} + \frac{p_2-p_1}{b_2-b_1}-\frac{p_1}{b_1} +\frac{p_2}{b_2-b_1} = \frac{2p_2b_1-2p_1b_2}{b_1(b_2-b_1)}.
$$ 
This is zero for $p_1=\frac{b_1}{b2}p_2$, which is the maximizing price. But then, $\frac{p_1}{b_1}=\frac{p_2}{b_2}=\frac{p_2-p_1}{b_2-b_1}$, which means that item 1 is never bought. This contradicts our assumption that we started out with a pricing with a minimum number of finite prices. We conclude that in this setting, there is a single finite price that is optimal against $\com$.
\end{proof} 

\subsection{Beyond MHR}
\label{appx:beyond-MHR}
\subsubsection{Observations related to the Myerson optimal revenue}\label{myerson-obs}

\paragraph{Myerson revenue.} 

Given an item $i$ with marginal distribution $F_i$, let $\Mye(F_i)$ denote the \emph{Myerson revenue} from this item, i.e., the optimal expected revenue from selling it to a buyer whose value is drawn from $F_i$. The Myerson revenue is obtained by pricing the item at its \emph{monopoly price}. 

\begin{observation}[Simple upper bound]
	\label{obs:simple-UB}
	Consider an instance of the max-min pricing problem with $n$ items and corresponding Myerson revenues $\Mye(F_1),\dots,\Mye(F_n)$.
	For every subset $S\subseteq [n]$ of the items, every pricing $p$ and every compatible distribution $F$, the contribution to the expected revenue $R(p,F)$ from selling items in $S$ is at most $\sum_{i\in S} \Mye(F_i)$. 
\end{observation}

Before proving \cref{obs:simple-UB} we derive two immediate corollaries:

\begin{corollary}[Upper bound on $R^*$]
\label{cor:simple-UB}
For every instance of the max-min pricing problem with $n$ items and corresponding Myerson revenues $\Mye(F_1),\dots,\Mye(F_n)$, 
$$
R^*\leq \sum_{i\in[n]} \Mye(F_i).
$$
\end{corollary}

\begin{corollary}
	\label{cor:simple-UB-subset}
	For every instance of the max-min pricing problem with $n$ items and corresponding Myerson revenues $\Mye(F_1),\dots,\Mye(F_n)$, for every pricing $p$ and every compatible distribution $F$, denote by $S$ the set of items with finite prices according to $p$. Then
	$$
	R(p,F)\leq \sum_{i\in S} \Mye(F_i).
	$$
\end{corollary}

\begin{proof}[Proof of \cref{obs:simple-UB}]
	For any compatible distribution $F$, an upper bound on the contribution to $R(p,F)$ from selling items in $S$ is the sum of expected revenues from selling each item $j\in S$ at price $p_j$ while all other prices are set to $\infty$, since this avoids any revenue loss due to cannibalization. The expected revenue from item $j$ is upper-bounded by $\Mye(F_j)$. The observation follows.
\end{proof}

Combining the next observation with normalization shows that up to log factors, we may restrict attention to settings with marginals whose Myerson revenues belong to the range $[1,2]$, as is indeed the case in \cref{ex:single-nonopt-cont}. 

\begin{observation}[Near-equal Myerson revenues]
	\label{obs:same-monopoly-rev}
	Consider an instance of the max-min pricing problem with finite Myerson revenues $\Mye(F_1)\ge \dots\ge \Mye(F_n)$. There exists a subset $S$ of items with Myerson revenues at most a factor 2 apart,
	and a pricing $p$ with prices of all items but those in $S$ set to $\infty$, such that $p$ is $O(\log n)$-max-min optimal. 
\end{observation}

\begin{proof}
	First observe that if we price at $\infty$ any item $j$ for which $\Mye(F_j)<\frac{\Mye(F_1)}{n}$, we lose at most a factor of 2 compared to the max-min optimal pricing. This is because the contribution of this set of items to $R^*$ is at most $\Mye(F_1)$ (\cref{obs:simple-UB}), and after ``throwing away'' these items by pricing them at $\infty$, the new robust revenue guarantee is still $\ge \Mye(F_1)$ (we can always set a single monopoly price for item 1). 
	We can now partition the remaining items into $\log n$ buckets such that in every bucket all items have the same Myerson revenue up to a factor of 2. Observe that the robust revenue guarantee for all items is upper-bounded by the sum of robust revenue guarantees for the items in each bucket separately. Thus pricing only the items in the best bucket at finite prices loses at most a factor of $\log n$, completing the proof.\footnote{We remark that instead of using $\infty$-prices we can price all items except for those in the best bucket at the highest price of the max-min optimal pricing for the best bucket.}
\end{proof}

\subsubsection{Analyzing Example~\ref{ex:single-nonopt-cont}}

The proofs below use definitions and observations from \cref{myerson-obs}.

\begin{proof}[Proof of \cref{lem:half-thresh}]
First, by \cref{cor:simple-UB-subset} and since $\Mye(\eqrev)=1$, for every pricing $p'$ with at most $k=|S|$ finite prices (and the rest of the prices set to $\infty$),
$R(p')\le k.$ 
In particular this holds for pricing $p$ as defined in \cref{lem:half-thresh}.
We argue that to complete the proof, it is sufficient to show that $R(p)=\Omega(k)$, from which we can conclude that $R(p)=\Theta(k)$. Indeed, by applying this to $S=[n]$ (the set of all items), we get that $R^*=\Omega(n)$, and by \cref{cor:simple-UB} we have $R^*=\Theta(n)$. So $p$ is $\Theta(\frac{n}{k})$-max-min optimal as required. 

It remains to show that $R(p)=\Omega(k)$. For every item $j\in S$, since the price is $p_j=2^j$ and the marginal distribution of $j$ is truncated at $t_j=2^{j+1}$, the maximum utility of the buyer from buying item~$j$ is $2^{j}$. 
Define a \emph{threshold} $\tau_j=\frac{3}{2}2^j$ for every item $j\in S$ (the threshold is halfway between this item's price and truncation point), and $\tau_j=\infty$ for every $j\notin S$.
Notice that $1-F_j(\tau_j)=\frac{2}{3}2^{-j}$ for every $j\in S$.
We claim that pointwise for every valuation profile~$v$ the following holds: If there exists some item $i\in[n]$ whose value \emph{clears} its threshold, that is, $v_{i}>\tau_{i}$, then the buyer purchases the highest-priced such item, i.e., item $j=\arg\max_{i\in[n]}\{\mathbbm{1}_{v_i> \tau_i}\cdot p_i\}$, or some other item with price at least $p_j$. 
This is because if $v_j>\tau_j$ then for every lower-priced item $j'<j$, the utility from buying $j'$ (upper-bounded by $2^{j'}\le 2^{j-1}$) is strictly lower than the utility $v_j-p_j$ from buying $j$; indeed, $v_j-p_j > \tau_j-p_j = 2^{j-1}$.

Consider any compatible joint distribution $F$. So far we have shown that $R(p,F)$ is at least the expected revenue from selling the highest-priced item that clears its threshold at every valuation profile $v\sim F$. Thus we can write:
\begin{align}
	R(p,F) &\ge 
	\sum_{j=1}^{n} \Pr_{v\sim F}[v_j>\tau_j\text{ and }\forall j'>j : v_{j'}\le \tau_{j'}]\cdot p_j\nonumber\\
	&=\sum_{j=1}^{n} \left(\Pr_{v\sim F}[\exists j'\ge j : v_{j'}> \tau_{j'}] - \Pr_{v\sim F}[\exists j'\ge j+1 : v_{j'}> \tau_{j'}]\right)\cdot p_j\nonumber\\
	&=
	\sum_{j=1}^{n} \left(\Pr_{v\sim F}[\exists j'\ge j : v_{j'}> \tau_{j'}]\right)\cdot(p_j-p_{j-1})\label{eq:rearrange}\\
	&\ge
	\sum_{j=1}^{n} \max_{j'\ge j}\bigl\{1-F_{j'}(\tau_{j'})\bigr\}\cdot 2^{j-1}\label{eq:ineq}\\ 
	&\ge \sum_{j\in S} \bigl(1-F_{j}(\tau_{j})\bigr)\cdot 2^{j-1} = \sum_{j\in S} \frac{2}{3}\cdot2^{-j}\cdot 2^{j-1} = \sum_{j\in S}\frac{1}{3}=\frac{|S|}{3},\label{eq:final}
\end{align}
where \cref{eq:rearrange} is by rearranging and defining $p_0=0$, \cref{eq:ineq} uses that $p_j=2^j$ and that $F$ is a compatible distribution, and \cref{eq:final} uses that $\max_{j'\ge j}\{1-F_{j'}(\tau_{j'})\}=1-F_{j}(\tau_{j})$ for every $j\in S$. 
We have shown that $R(p,F)=\Omega\bigl(|S|\bigr)$ for every compatible $F$, and this means that $R(p)=\Omega\bigl(|S|\bigr)$, completing the proof.
\end{proof}

\begin{proof}[Proof of \cref{pro:single-nonopt-cont}]
Consider \cref{ex:single-nonopt-cont}. 
It is sufficient to show that with $\le k$ different prices, $R(p)\le k$. 
The proof then follows by \cref{lem:half-thresh}.
We achieve this by showing that $R(p,\com)\le k$. 

Assume consistent tie-breaking according to price (and then possibly according to index, etc.). For every pricing $p$ with $\le k$ different prices, we show a pricing $p'$ with $\le k$ \emph{finite} prices (and the rest of the prices set to $\infty$) such that $R(p,\com)\le R(p',\com)$. We define $p'$ as follows: For every distinct price $p_\ell$ in $p$ where $\ell\in[k]$, let $j_\ell$ be the item with this price that has the highest truncation point. Let $p'_{j_\ell}=p_\ell$, and for every other item $i$ with price $p_\ell$ in the pricing $p$, set $p'_i=\infty$. 

We now compare the expected revenues $R(p,\com)$ and $R(p',\com)$: Fix a valuation profile~$v$ in the support of $\com$ for which an item is purchased given~$p$. 
Denote the price at which the item is purchased by $p_\ell$. Clearly this item $j^*$ has the maximum utility $u^*$ among all other items given~$p$. We argue that the utility of item $j_\ell$ given $p'$ is also $u^*$. This is because by definition of $\com$, valuation profile $v$ corresponds to a certain quantile $q$, i.e., $v_j=v_j(q)$ for every item~$j$. Since $j_\ell$ has a (weakly) higher truncation point than $j^*$, we have that $v_{j_\ell}(q)\ge v_{j^*}(q)$. Since $p'_{j_\ell}=p_{j_\ell}=p_{j^*}=p_\ell$, the utility of item $j_\ell$ given $p'$ is the same as given $p$ and both are at least $u^*$. Since $u^*$ is the maximum utility given $p$, the argument is complete.

To complete the comparison, notice that the utilities of all other items besides $j_\ell$ given $p'$ are a subset of the utilities of all other items given $p$. Thus $u^*$ is the maximum utility given $p'$ as well, and assuming consistent tie-breaking according to price, we conclude that the item purchased given $p'$ has price~$p_\ell$. Therefore, $R(p,\com)\le R(p',\com)$. The latter is upper-bounded by $k$ according to \cref{cor:simple-UB-subset}, completing the proof.
\end{proof}

\end{document}